\documentclass[a4paper,11pt]{article}
\usepackage[utf8]{inputenc}
\usepackage[english]{babel}
\usepackage{amssymb}
\usepackage{amsmath}
\usepackage{amsthm}
\usepackage{graphicx}
\usepackage{fullpage}
\usepackage{authblk}
\usepackage{pgfplots}
\usepackage[caption=false]{subfig}
\usepackage[colorlinks,bookmarksopen,bookmarksnumbered,citecolor=blue,linkcolor=magenta]{hyperref}

\newcommand{\myqed}{}

\pgfplotsset{compat=newest}
\usetikzlibrary{positioning, decorations.text, decorations.pathmorphing, calc}

\newcommand{\Q}{\mathbb{Q}}
\newcommand{\R}{\mathbb{R}}
\newcommand{\Z}{\mathbb{Z}}

\newcommand{\TGr}{{\rm \tilde{Gr}}}
\newcommand{\Gr}{{\rm Gr}}
\newcommand{\TAr}{{\rm \tilde{N}}}
\newcommand{\Ar}{{\rm N}}
\newcommand{\Area}{{\rm Area}}

\newcommand{\OPT}{\ensuremath{\mathop{\mathrm{OPT}}\nolimits}}
\newcommand{\dist}{\ensuremath{\mathop{\mathrm{dist}}\nolimits}}

\newtheorem{theorem}{Theorem}
\newtheorem{definition}{Definition}
\newtheorem{lemma}{Lemma}
\newtheorem{corollary}{Corollary}
\newtheorem{claim}{Claim}

\begin{document}

\title{Polynomial-Time Approximation Schemes for\\
Circle and Other Packing Problems%
\thanks{This work was partially supported by CNPq (grants \mbox{303987/2010-3},
\mbox{306860/2010-4}, \mbox{477203/2012-4}, and \mbox{477692/2012-5}), FAPESP
(grants \mbox{2010/20710-4}, \mbox{2013/02434-8}, \mbox{2013/03447-6}, and
\mbox{2013/21744-8}), and Project MaClinC of NUMEC at USP, Brazil.}}

\author[1]{Fl\'avio K.\ Miyazawa}
\author[2]{Lehilton L.\ C.\ Pedrosa}
\author[1]{Rafael C.\ S.\ Schouery}
\author[3]{\authorcr Maxim Sviridenko}
\author[2]{Yoshiko Wakabayashi}
\setlength{\affilsep}{1cm}
\affil[1]{Institute of Computing, University of Campinas, Brazil \authorcr \texttt{\small \{fkm,schouery\}@ic.unicamp.br}\medskip}
\affil[2]{Institute of Mathematics and Statistics, University of S\~ao Paulo, Brazil\authorcr \texttt{\small  \{lehilton,yw\}@ime.usp.br}\medskip}
\affil[3]{Yahoo!\ Labs, New York, NY\authorcr \texttt{\small sviri@yahoo-inc.com}}

\date{}
\maketitle

\begin{abstract}
We give an asymptotic approximation scheme (APTAS) for the problem of packing a
set of circles into a minimum number of unit square bins. To obtain rational
solutions, we use augmented bins of height $1+\gamma$, for some arbitrarily
small number $\gamma > 0$. Our algorithm is polynomial on $\log 1/\gamma$, and
thus $\gamma$ is part of the problem input. For the special case that $\gamma$
is constant, we give a (one dimensional) resource augmentation scheme, that is,
we obtain a packing into bins of unit width and height $1+\gamma$ using no more
than the number of bins in an optimal packing. Additionally, we obtain an APTAS
for the circle strip packing problem, whose goal is to pack a set of circles
into a strip of unit width and minimum height. These are the first approximation
and resource augmentation schemes for these problems.

Our algorithm is based on novel ideas of iteratively separating small and large
items, and may be extended to a wide range of packing problems that satisfy
certain conditions. These extensions comprise problems with different kinds of
items, such as regular polygons, or with bins of different shapes, such as
circles and spheres. As an example, we obtain APTAS's for the problems of
packing $d$-dimensional spheres into hypercubes under the $L_p$-norm.
\end{abstract}

\section{Introduction}

In the \emph{circle bin packing problem}, we are given a list $\mathcal{C}$ of
$n$ circles identified by their indexes, ${\mathcal{C} =\{1,\ldots,n\}}$, where
circle $i$ has radius $r_i \le 1/2$, for $1 \le i \le n$, and an unlimited
number of identical square bins of unit side. A packing is a non-overlapping
placement of circles into a set of bins, such that every circle is fully
contained in a bin. Here, one circle $i$ is packed into one bin if we can
associate coordinates $x_i,y_i$ to its center, such that $r_i \le x_i, y_i \le 1
- r_i$, and for every other circle $j$ with coordinates $x_j,y_j$ that is packed
into the same bin, we have $(x_i - x_j)^2 + (y_i - y_j)^2 \ge (r_i + r_j)^2$.
The objective is to find a packing of $\mathcal{C}$ into a minimum number of
bins. In the \emph{circle strip packing problem}, that is also known as
\emph{circle two-dimensional open dimension
problem}~\cite{HifiM09,WaescherHS2007}, the set of circles $\mathcal{C}$ must be
packed into a strip of unit width and unbounded height, and the objective is to
obtain a packing of minimum height.

There are several results in the literature for packing problems involving
circles, that are tackled using different methods, such as continuous and
nonlinear systems, and discrete methods~\cite{BirginG10}. A large portion of
the works deals with the problem of maximizing the number of circles of a given
radius in a square~\cite{SzaboMCSCG07}. In the context of origami design,
Demaine~\emph{et~al.}~\cite{DemaineFL10} proved that the decision problem that
asks whether a set of circles can be packed into a unit square or into an
equilateral triangle is NP-hard. Therefore, the circle bin packing problem and
the circle strip packing problem are also NP-hard.

We are interested in the design of approximation algorithms for the circle bin
packing and the circle strip packing problems.  As it is usual for packing
problems, the measure we look for is the asymptotic performance. Given an
algorithm~$\mathcal{A}$, and a problem instance~$I$, we denote by
$\mathcal{A}(I)$ the value of the solution produced by $\mathcal{A}$, and
by~$\OPT(I)$ the value of an optimum solution for~$I$. For some $\alpha \ge 1$,
we say that a polynomial-time algorithm $\mathcal{A}$ (for a minimization
problem) is an \emph{asymptotic $\alpha$-approximation algorithm} if, for every
instance $I$, we obtain $\mathcal{A}(I) \le \alpha \OPT(I) + O(1)$. Also, a
family of polynomial-time algorithms $\{\mathcal{A_\varepsilon}\}$ is said to be
an \emph{asymptotic polynomial-time approximation scheme} (APTAS) if, for every
instance $I$, and fixed $\varepsilon > 0$, we have $\mathcal{A_\varepsilon}(I)
\le (1+\nobreak\varepsilon) \OPT(I) + O(1)$. If the constant term $O(1)$ is
omitted from the definitions, that we say that $\mathcal{A}$ is an
\emph{$\alpha$-approximation algorithm}, and $\{\mathcal{A_\varepsilon}\}$ is a
\emph{polynomial-time approximation scheme} (PTAS), respectively.

\subsection{Our results and techniques}

In this paper, we present a new algorithm for a series of packing problems. We
give APTAS's for both the circle bin packing, and the circle strip packing
problems. The bin packing problem is considered when we allow the use of
augmented bins of unit width and height $1+\gamma$, for some arbitrarily small
$\gamma> 0$. This relaxation is necessary due to numeric concerns, as the
coordinates for the circles' centers obtained by our algorithms are given by
roots of polynomial equations, that are possibly irrational. In order to obtain
rational solutions, we approximate the coordinates and increase the height of a
bin slightly. To the best of our knowledge, these are the first approximation
guarantees for these natural problems.

We highlight that the time complexity of the algorithm depends polynomially on
$\log 1/\gamma$, so the value of parameter $\gamma$ may be given as part of the
problem instance. For the case that a bin is enlarged by an arbitrary but
constant value $\varepsilon > 0$, we give a resource augmentation scheme in one
dimension, that is, for any constant $\varepsilon > 0$, we develop a
polynomial-time algorithm $\mathcal{A}_\varepsilon$ that returns a packing of
$\mathcal{C}$ into bins of unit width and height $1+\varepsilon$ with size
$\mathcal{A}_\varepsilon(\mathcal{C}) \le \OPT(\mathcal{C})$, where
$\OPT(\mathcal{C})$ is the optimal value of the problem without resource
augmentation.

Although the algorithm presented here is described only for the circle bin
packing, it can be seen as a \emph{unified framework}, and extends to a wide
range of different packing problems, such as the bin packing of ellipses,
regular polygons, and many others. Indeed, as an illustrative example, we show
how to generalize our results to the bin packing and to the strip packing of
$d$-dimensional $L_p$-norm spheres.

Our algorithm uses some techniques that have appeared in the literature in
several and new interesting ways. As usual in the packing of rectangular items,
our algorithm distinguishes between ``large'' and ``small'' items. However, this
distinction is dynamic, so that an item can be considered small in one
iteration, but large in a later one. There are two main novel ideas that differ
from the approaches for rectangle packing, and that are needed for the circle
packing. First, we reduce the packing of large circles to the problem of solving
a semi-algebraic system, that can be done with the aid of standard quantifier
elimination algorithms from algebra. This helped us to avoid the use of
combinatorial brute-force algorithms that are based on discretization, which
would depend exponentially on $\log 1/\gamma$. Second, to pack small items, we
cut the free space of previous packings into smaller sub-bins, and use the
algorithm for large items, recursively.

\subsection{Related works}

In the literature of approximation algorithms, the majority of the works
consider the packing of simple items into larger recipients, such as rectangular
bins and strips. Most of the works which give approximation guarantees are
interested in rectangular items or $d$-dimensional boxes. The packing problems
involving circles are mainly considered through heuristics, or numerical
methods, and, to our knowledge, there is no approximation algorithm for the
circle bin packing or for the circle strip packing problems.
On the practical side, packing problems have numerous applications, such as
packaging of boxes in containers, or cutting of material. An application of
circular packing is, for example, obtaining a maximal coverage of radio towers
in a geographical region~\cite{SzaboMCSCG07}.

The problem of finding the densest packing of equal circles into a square has
been largely investigated using many different optimization methods. For an
extensive book on this problem, and corresponding approaches,
see~\cite{SzaboMCSCG07}. The case of circles of different sizes is considered
in~\cite{GeorgeGL95}, where heuristics, such as genetic algorithms, are proposed
to pack circles into a rectangular container. The circle strip packing has been
considered using many approaches, such as branch-and-bound, metaheuristics, etc.
For a broad list of algorithms for the circle strip packing, and related circle
packing problems, see~\cite{HifiM09} and references therein.

For the problem of packing rectangles into rectangular bins,
Chung~\emph{et~al.}~\cite{ChungGJ82} presented a hybrid algorithm, called HFF,
combining a one-dimensional packing algorithm~(FFD) with a strip packing
algorithm (FFDH) to obtain an algorithm with asymptotic approximation
ratio~$2.125$. Caprara~\cite{Caprara02} proved that the asymptotic approximation
ratio of the algorithm HFF is~$2.077$, and also presented a better algorithm,
with approximation ratio $1.691$. Bansal~\emph{et~al.}~\cite{BansalCS10}
improved this ratio with a probabilistic algorithm, that can be derandomized,
with asymptotic approximation ratio that can be as close to $1.525$ as desired.
Recently, Bansal and Khan~\cite{BansalK14} gave an asymptotic
$1.405$-approximation. Considering non-asymptotic approximation
ratio, Harren and van Stee~\cite{HarrenS12} showed that HFF has ratio~$3$, and
presented an algorithm with ratio~$2$~\cite{HarrenS09}.
For the bin packing of $d$-dimensional cubes,
Kohayakawa~\emph{et~al.}~\cite{KohayakawaMRW04} showed an asymptotic ratio of
$2-(2/3)^d$, later improved to an APTAS by
Bansal~\emph{et~al.}~\cite{BansalCKS06}. For a survey on bin packing,
see~\cite{CoffmanCGMV13}.

The first approximation algorithm for the rectangle strip packing problem was
proposed by Baker~\emph{et~al.}~\cite{BakerCR80}. They presented the so called
BL (Bottom-Leftmost) algorithm, and showed that it has approximation ratio $3$.
For the special case when all items are squares, the approximation ratio of BL
algorithm is at most $2$. Coffman~\emph{et~al.}~\cite{CoffmanGJT80} presented
three algorithms, denoted by NFDH (Next Fit Decreasing Height), FFDH (First Fit
Decreasing Height), and SF (Split Fit), with asymptotic approximation ratios of
$2$, $1.7$, and $1.5$, respectively. They also showed that, when the items are
squares, FFDH has an asymptotic approximation ratio of $1.5$, and, when all items
have width at most $1/m$, the algorithms FFDH and SFFDH have asymptotic
approximation ratios of $(m+1)/m$ and $(m+2)/(m+1)$, respectively. The best
known ratio for the problem was obtained by Kenyon and R\'emila~\cite{KenyonR00}, who
presented an asymptotic approximation scheme.
Considering non-asymptotic approximation algorithms, Sleator~\cite{Sleator80}
presented a ratio $2.5$. This result was improved to $2$ independently by
Schiermeyer~\cite{Schiermeyer94} and Steinberg~\cite{Steinberg97}, then to
$1.9396$ by Harren and van Stee~\cite{HarrenS09}, and finally to
$5/3+\varepsilon$ by Harren~\emph{et~al.}~\cite{HarrenJPS11}.

For the 3-dimensional strip packing problem, whose items are boxes, Li and
Cheng~\cite{LiC90} were the first to present an asymptotic $3.25$-approximation
algorithm. Their algorithm was shown to have approximation ratio
$2.67$~\cite{MiyazawaW97}, $2+\varepsilon$~\cite{JansenS06} and finally
$1.69$~\cite{BansalHISZ13}. Bansal \emph{et~al.}~\cite{BansalCKS06} showed that
there is no asymptotic approximation scheme for the rectangle bin packing
problem, which implies that there is no APTAS for the 3-dimensional strip
packing problem. When the items are cubes, the first specialized algorithm was
shown to have asymptotic ratio of $2.6875$~\cite{LiC90}, and the best result is
an asymptotic approximation scheme due to
Bansal~\emph{et~al.}~\cite{BansalHISZ13}.

\paragraph{Organization}

The remainder of this paper is organized as follows. In
Section~\ref{sec:algebraic}, we discuss how to decide whether a set of~$n$
circles can all be packed in a rectangular bin using algebraic quantifier
elimination. In Section~\ref{sec:large}, we give approximation algorithms for
the case of ``large'' circles. In Section~\ref{sec:aptas}, we present APTAS's
for the circle bin packing problem, and for the circle strip packing problem. In
Section~\ref{sec:resource}, we give a resource augmentation scheme for the
circle bin packing. In Section~\ref{sec:generalizations}, we generalize the
results of Section~\ref{sec:aptas} to the case of multidimensional spheres, and
to the case of items and bins with different shapes. In
Section~\ref{sec:conclusion}, we give the final remarks.

\section{Circle Packing Through Algebraic Quantifier Elimination}\label{sec:algebraic}

Throughout this paper, we will consider instances for circle packing problems as
in the following definition.

\begin{definition}
A triple $(\mathcal{C}, w, h)$ is an instance for the circle packing problem if
$h, w \in \Q_+$, and $\mathcal{C}~=~\{1, \ldots, n\}$ is a set of circles, such
that each circle $i$ has radius $r_i \in \Q_+$, with $2r_i \le \min\{w,h\}$, for
$1 \le i \le n$.
\end{definition}

In this section, we consider the following \emph{circle packing decision
problem}. An instance for this problem is a triple~${(\mathcal{C}, w, h)}$, and
the objective is to decide whether the circles in $\mathcal{C}$ can be packed
into a bin of size $w \times h$, that is, a rectangle of width~$w$ and
height~$h$. In the case of a positive answer, a realization of the packing
should also be returned. More precisely, for each circle $i$, with $1 \le i \le
n$, we want to find a point $(x_i, y_i) \in \R_+^2$ that represents the center
of $i$ in a rectangle whose bottom-left and top-right corners correspond to
points $(0,0)$ and $(w,h)$, respectively.

The circle packing decision problem can be equivalently formulated as deciding
whether there are real numbers $x_i, y_i \in \R_+$, for $1 \le i \le n$, that
satisfy the constraints
\begin{eqnarray}
 & (x_i - x_j)^2 + (y_i - y_j)^2 \ge (r_i + r_j)^2 & \mbox{ for } 1 \le i < j \le n \label{eq1},\\
 & r_i \le x_i  \le w - r_i & \mbox{ for } 1 \le i \le n, \mbox{ and } \label{eq2}\\
 & r_i \le y_i  \le h - r_i & \mbox{ for } 1 \le i \le n.\label{eq3}
\end{eqnarray}
The set of constraints~\eqref{eq1} guarantees that no two circles intersect, and
the sets of constraints~\eqref{eq2} and~\eqref{eq3} ensure that each circle has
to be packed entirely in the rectangle that expands from the origin~$(0,0)$ to
the point $(w, h)$.

We observe that the set of solutions that satisfy~\eqref{eq1}-\eqref{eq3} is a
semi-algebraic set in the field of the real numbers. Thus, the circle packing
decision problem corresponds to deciding whether this semi-algebraic set is not
empty. We also can rewrite the constraints in~\eqref{eq1}-\eqref{eq3} as
$f_i(x_1, y_1, ..., x_n, y_n) \ge 0$, for $1 \le i \le s$, where $s$ is the
total number of constraints, and $f_i \in \Q[x_1, y_1, ..., x_n, y_n]$ is a
polynomial with rational coefficients. Then, the circle packing problem is
equivalent to deciding the truth of the formula
\begin{equation}
\textstyle
(\exists x_1) (\exists y_1) \dots (\exists x_n) (\exists y_n)
\bigwedge_{i = 0}^s f_i(x_1, y_1, ..., x_n, y_n) \ge 0.\label{eq-formula}
\end{equation}

We can use any algorithm for the more general quantifier elimination problem to
decide this formula. There are several algorithms for this problem, such as the
algorithm of \mbox{Tarski-Seidenberg} Theorem~\cite{Tarski51}, that is not elementary
recursive, or the Cylindrical Decomposition Algorithm~\cite{Collins75}, that is
doubly exponential in the number of variables. Since the formula corresponding
to the circle packing problem contains only one block of variables (of
existential quantifiers), we can use faster algorithms for the corresponding
algebraic existential problem, such as the algorithms of Grigor'ev and
Vorobjov~\cite{GrigorevV88}, or of Basu~\emph{et~al.}~\cite{BasuPR96}. For an
extensive list of algorithms for real algebraic geometry, see~\cite{BasuPR06}.

\paragraph{Sampling points of the solution}
Any of the algorithms above receiving formula~\eqref{eq-formula} as input will
return ``true'' if, and only if, the circles in $\mathcal{C}$ can be packed into
a bin of size~${w \times h}$. When the answer is ``true'', we are also
interested in a realization of such packing. The algorithms
in~\cite{GrigorevV88,BasuPR96} are based on critical points, that is, they also
return a finite set of points that meets every semi-algebraic connected
component of the semi-algebraic set. Thus, a realization of the packing can be
obtained by choosing one of such points (that is a point that corresponds to a
connected component where all polynomials $f_i$, $1 \le i \le s$, are
nonnegative).

Typically, a sample point is represented by a tuple $(f(x), g_0(x), \dots,
g_k(x))$ of $k+2$ univariate polynomials with coefficients in~$\Q$, where $k$ is
the number of variables, and the value of the $i$-th variable is $g_i(x)/g_0(x)$
evaluated at a real root of $f(x)$. Since a point in a semi-algebraic set
may potentially be irrational, we use the algorithm of Grigor'ev and
Vorobjov~\cite{GrigorevV88}, for which we have $g_0(x) = 1$, and thus an
approximate rational solution of arbitrary precision can be readily obtained. In
particular, the theorem given in~\cite{GrigorevV88} implies the following
result.

\begin{corollary}\label{theorem:polysystem}
Let $f_1, ..., f_s \in \Q[x_1,y_1, ..., x_n,y_n]$ be polynomials with
coefficients of bit-size at most~$m$, and maximum degree $2$. There is an
algorithm that decides the truth of formula~\eqref{eq-formula}, with running
time $m^{O(1)}s^{O(n^2)}$. In the case of affirmative answer, then the algorithm
also returns polynomials $f, g_1,h_1,\dots, g_n,h_n \in \Q[x]$ with coefficients
of bit-size at most $m^{O(1)}s^{O(n)}$, and maximum degree $s^{O(n)}$, such that
for a root $x$ of $f(x)$, the assignment $x_1 = g_1(x), y_1 = h_1(x),
..., x_n =g_n(x), y_n =h_n(x)$ is a realization of~\eqref{eq-formula}.
Moreover, for any rational $\alpha > 0$, we can obtain $x'_1, y'_1, \dots, x'_n,
y'_n \in \Q$, such that $|x'_i - x_i| \le \alpha$ and $|y'_i - y_i| \le \alpha$,
$1 \le i \le n$, with running time~$(\log (1/\alpha)m)^{O(1)}s^{O(n^2)}$.
\end{corollary}

\section{Approximate Bin Packing of Large Circles}\label{sec:large}

In this section, we consider the special case of the circle bin packing problem
in which the minimum radius of a circle is at least a constant. For this case,
the maximum number of circles that fit in a bin is constant, so we can use the
algorithm of Corollary~\ref{theorem:polysystem} to decide in constant time
whether a given set of circles can be packed into a bin. Since
Corollary~\ref{theorem:polysystem} only gives us rational solutions that are
close to real-valued packings, we will first transform an approximate packing
into a non-overlapping packing in an augmented bin.

Later, we obtain a PTAS for the special case of the circle bin packing problem with
large items.

\subsection{Transforming approximate packings}

We start with the next definition to deal with approximate circle bin packings.
In the following, we denote by $\dist(p,q)$ the Euclidean distance between two
points $p, q$ of the $2$-dimensional space.

\begin{definition}
Given an instance $(\mathcal{C}, w, h)$ for the circle packing problem, and a
number $\varepsilon \ge 0$, we say that a set of points $p_i = (x_i, y_i)$, for
${1 \leq i \leq n}$, is an \emph{$\mathbf{\varepsilon}$-packing} of
$\mathcal{C}$ into a rectangular bin of size~$w \times h$, if the following
hold:
\begin{eqnarray}
& \dist(p_i,p_j)    \geq r_i + r_j - \varepsilon \ge 0  & \mbox{ for } 1 \le i < j \le n,\\
& r_i - \varepsilon \leq x_i \leq w - r_i + \varepsilon & \mbox{ for } 1 \le i \le n,\mbox{ and}\\
& r_i - \varepsilon \leq y_i \leq h - r_i + \varepsilon & \mbox{ for } 1 \le i \le n.
\end{eqnarray}
\end{definition}

We adopt the following strategy to fix intersections of an approximate bin
packing: (a) first, we shift the $x$-coordinate of all circles that intersect
the left or right border until they are fully contained in the bin, (b) then, we
iteratively lift each circle in order of the $y$-coordinate by an appropriate
distance so that it does not intersect any of the circles considered in previous
iterations. First, the next lemma bounds the distance that one circle needs to
be raised to avoid intersection with lower circles in an $\varepsilon
h$-packing. Then, Lemma~\ref{lemma:augmentation} transforms an approximate
packing into a packing with bins of augmented height.

\begin{lemma}\label{lemma:rising}
Let $r_1, r_2, h, \varepsilon$ be positive numbers such that $\varepsilon h \leq
r_1 + r_2 \leq h$, and ${p_1 = (x_1,y_1)}$, ${p_2 = (x_2,y_2)}$ be points in
$\R^2$. If ${y_1 \geq y_2}$, ${\dist(p_1, p_2) \geq r_1 + r_2 - \varepsilon h}$,
and ${p_1' = (x_1, y_1+\sqrt{2\varepsilon} h)}$, then ${\dist(p_1',p_2) \geq r_1
+ r_2}$.
\end{lemma}

\begin{proof}
By direct calculation,
\begin{eqnarray*}
\dist(p_1', p_2)
  &=&    \sqrt{(x_1 - x_2)^2 + (y_1 + \sqrt{2\varepsilon} h - y_2)^2} \\
  &=&    \sqrt{(x_1 - x_2)^2 + (y_1 - y_2)^2 + 2\sqrt{2\varepsilon} h(y_1 - y_2) + 2\varepsilon h^2}\\
  &=&    \sqrt{\dist(p_1,p_2)^2 + 2\sqrt{2\varepsilon} h(y_1 - y_2) +2\varepsilon h^2}\\
  &\geq& \sqrt{(r_1 + r_2 - \varepsilon h)^2 + 2\varepsilon h^2}\\
  &=&    \sqrt{(r_1 + r_2)^2 - 2\varepsilon h(r_1 + r_2) + \varepsilon^2 h^2 + 2\varepsilon h^2}\\
  &\geq& r_1 + r_2,
\end{eqnarray*}
where the last inequality follows from $r_1  + r_2 \leq h$.
\myqed\end{proof}

\begin{lemma}\label{lemma:augmentation}
Given an instance $(\mathcal{C}, w, h)$ for the circle packing problem, and a
corresponding $\varepsilon h$-packing of $\mathcal{C}$ into a bin of size~$w
\times h$ for some $\varepsilon > 0$, we can find a packing of $\mathcal{C}$
into a bin of size~$w \times (1 + n\sqrt{6\varepsilon})h$ in linear time.
\end{lemma}

\begin{proof}
For $1 \leq i \leq n$, let $p_i = (x_i, y_i)$ be the center of circle $i$
corresponding to the $\varepsilon h$-packing. We start by modifying the given
$\varepsilon h$-packing to obtain a $3\varepsilon h$-packing into a bin of
size~$w \times (h+2\varepsilon h)$, with the additional property that no circle
intersects a border of such rectangular bin.
For each~$1 \leq i \leq n$, let $p_i' = (x_i', y_i')$ be the center of circle
$i$ in the modified packing. The $y$-coordinate is defined as $y_i' = y_i +
\varepsilon$, and the $x$-coordinate is defined as: $x_i' = x_i$ if $i$ does
not intersect the left or right border; $x_i' = r_i$ if $i$ intersects the
left border; and $x_i' = w - r_i$ if $i$ intersects the right border (notice
that $i$ cannot intersect both the left and the right borders, since $2 r_i
\le w$). Clearly, the definition of the centers guarantees that no circle
intersects any border of the augmented bin. To see that the set of points $p_i'$
is a $3\varepsilon h$-packing, just note that any two circles $i, j$ are
lifted by the same distance, so by the triangle inequality $\dist(p_i',p_j') \ge
\dist(p_i,p_j) - 2 \varepsilon h \ge r_i + r_j - 3 \varepsilon h$.

Now, we transform the $3\varepsilon h$-packing into a packing in a bin of
width~$w$ and height ${h + 2\varepsilon h+ (n-1)\sqrt{6\varepsilon}h \le (1 +
n\sqrt{6\varepsilon})h}$. We can assume, without loss of generality, that
circles $1, \ldots, n$ are ordered in nondecreasing order of the $y$-coordinate
of their centers $p_i'$. For every circle~$i$,~$1 \le i \le n$, define its new
center as $p''_i = (x_i', y_i' + (i-1)\sqrt{6\varepsilon}h)$. Notice that the
$y$-coordinate of the last circle is increased by $(n-1)\sqrt{6\varepsilon}h$,
and thus no circle intersects any of the borders. To complete the proof, we show
that no two circles $i$ and $j$, with $1 \leq j < i \leq n$, intersect:
\begin{align*}
\dist(p''_i, p''_j)
  &=    \dist((x_i',y_i' + (i-1)\sqrt{6\varepsilon}h), (x_j', y_j' + (j-1)\sqrt{6\varepsilon}h)) \\
  &=    \dist((x_i',y_i' + (i-j)\sqrt{6\varepsilon}h),(x_j', y_j'))\\
  &\geq \dist((x_i',y_i' + \sqrt{6\varepsilon}h),(x_j', y_j')) \geq r_i + r_j,
\end{align*}
where the last inequality follows from Lemma~\ref{lemma:rising}, and the fact
that we had a $3\varepsilon h$-packing.
\myqed\end{proof}

\subsection{A PTAS for large items}

The optimal value for the circle bin packing problem is defined next.

\begin{definition}
Given an instance $(\mathcal{C}, w, h)$ for the circle bin packing, we denote by
$\OPT_{w\times h}(\mathcal{C})$ the minimum number of rectangular bins of
size~$w \times h$ that are necessary to pack~$\mathcal{C}$.
\end{definition}

We will obtain an approximation algorithm for the bin packing of large circles,
that is, assuming that the radius of each circle is greater than a given
constant. In this case, we can use an algorithm similar to that by Fernandez De La
Vega and Lueker~\cite{FernandezdelaVegaL81} to obtain a bin packing of large
circles. In this case, the maximum number of circles that fit in a bin is at
most a constant,~$M$, so we can partition the set of circles into a small number
of groups with approximate sizes, and enumerate all patterns of groups with no
more than $M$ circles. Then, we may apply the algorithm of
Corollary~\ref{theorem:polysystem} to list which patterns correspond to feasible
packings, and use integer programming in fixed dimension to find out how many
bins of each pattern are necessary to cover all circles.

We split the algorithm in two parts. First, Lemma~\ref{lemma:fixed_radii}
considers the special case where the number of different radii is bounded by a
constant, and obtain a bin packing using at most the optimal number of bins,
$\OPT_{w\times h}(\mathcal{C})$. Then, Theorem~\ref{theorem:packing_large}
reduces the general case to the case of bounded number of radii, and obtain a
bin packing using at most an additional $\varepsilon$ fraction of~$\OPT_{w\times
h}(\mathcal{C})$.
In the following, we will denote the area of the circle of radius $r$ by
$\Area(r)$.

\begin{lemma}\label{lemma:fixed_radii}
Let $(\mathcal{C}, w, h)$ be an instance of the circle bin packing, such that
$w, h \in O(1)$, $\min_{1\le i\le n} r_i \ge \delta$, and $|\{r_1, \dots, r_n\}|
\le K$, for constants $K$ and $\delta$.
For any given number $\gamma > 0$, we can obtain a packing of $\mathcal{C}$ into
at most $\OPT_{w\times h}(\mathcal{C})$ rectangular bins of size~$w \times
(1+\gamma)h$ in polynomial time.
\end{lemma}

\begin{proof}
Notice that a bin of size~$w \times h$ may contain at most $M = \lceil
wh/\Area(\delta) \rceil$ circles of~$\mathcal{C}$. Consider an ordering of
distinct radii $\bar{r}_1, \dots, \bar{r}_K$. We say that a vector of
nonnegative integers~${c = (c_1, \ldots, c_K)}$ with $\sum_{i =  1}^{K} c_i \leq
M$ is a configuration, and, for each $1 \le i \le K$, that $c_i$ is the number
of circles with radius $\bar{r}_i$ of $c$. A configuration $c$ is said to be
feasible if there is a packing into a bin of size~$w \times h$ containing all
circles of $c$.

Let $\varepsilon = \gamma^2/(6M^2)$, and let $\alpha = \varepsilon h/4$. We
enumerate each of the (at most $M^K$) configurations~$c$, and use the algorithm
of Corollary~\ref{theorem:polysystem} to decide whether $c$ is feasible. For
each feasible $c$, we also obtain an approximate packing, such that each circle
of $c$ has rational center $p'$ at distance at most $2 \alpha$ of the center $p$
in the packing realization. Then, for any two circles, with centers $p_1$ and
$p_2$ in the packing realization, and approximate rational centers $p_1'$ and
$p_2'$, we have $\dist(p_1,p_2) - \dist(p_1',p_2') \le 4 \alpha = \varepsilon
h$. Thus, the obtained approximate packing is an $\varepsilon h$-packing of
circles in $c$. For each feasible packing, we use
Lemma~\ref{lemma:augmentation}, and obtain a packing of the circles of $c$ into
a bin of width~$w$ and height $(1 + M \sqrt{6 \varepsilon})h = (1 + \gamma) h$.
Let $\mathcal{X}$ be the set of feasible configurations, and let $n_i$, for $1
\le i \le K$, denote the number of circles in $\mathcal{C}$ with radius
$\bar{r}_i$. Solving the following integer program, we obtain a bin packing of
size $\OPT_{w\times h}(\mathcal{C})$ that contains $x_c$ bins of configuration
$c$ for each $c \in \mathcal{X}$.
\begin{eqnarray*}
\mbox{minimize}   & \textstyle\sum_{c \in \mathcal{X}} x_c \\
\mbox{subject to} & \textstyle\sum_{c \in \mathcal{X}} c_i x_c \geq  n_i, & 1 \le i \le K\\
                  &                                        x_c \in \Z_+   & c \in \mathcal{X}
\end{eqnarray*}

Since this program has a constant number of variables, and bit-size
$O(\log n)$, it can be solved in $O(\log n)$ using fixed dimension
integer programming~\cite{Lenstra83,Eisenbrand03}.
\myqed\end{proof}

Now, the following theorem gives a $(1+\varepsilon)$-approximation for the
particular case in which circles' radii are greater than a constant.

\begin{theorem}\label{theorem:packing_large}
Let $(\mathcal{C}, w, h)$ be an instance of the circle bin packing, such that
$w, h \in O(1)$, and $\min_{1\le i\le n} r_i \ge \delta$, for some constant $\delta$.
For any given constant $\varepsilon > 0$, and number $\gamma > 0$, there is a
polynomial-time algorithm that packs $\mathcal{C}$ into at most
$(1+\varepsilon)\OPT_{w\times h}(\mathcal{C})$ rectangular bins of size~$w
\times (1+\gamma)h$.
\end{theorem}

\begin{proof}
First, let $K = \lceil 2/(\varepsilon \Area(\delta))\rceil$.
If $n \le K$, then we use the algorithm of Lemma~\ref{lemma:fixed_radii} on
instance $\mathcal{C}$, and obtain a packing of $\mathcal{C}$ into at most
$\OPT_{w\times h}(\mathcal{C})$ bins of size~$w \times (1+\gamma)h$. If $n > K$,
then we let $Q = \lfloor \varepsilon n \Area(\delta) \rfloor > 1$, and execute
the following steps:
\begin{enumerate}
  \item sort the circles in non-increasing order of radius;
  \item partition $\mathcal{C}$ in groups of up to $Q$ consecutive circles greedily;
  \item create an instance $\mathcal{C}'$ by changing the radius of each circle
        in $\mathcal{C}$ to the smallest radius of its group;
  \item use the algorithm of Lemma~\ref{lemma:fixed_radii} to find a
        packing of $\mathcal{C'}$.
\end{enumerate}
We have obtained a packing $P'$ of $\mathcal{C'}$ of size $\OPT_{w\times
h}(\mathcal{C'})$ into rectangular bins of size~$w \times (1+\gamma)h$. Notice
that, with exception of circles in the first group, every circle in
$\mathcal{C}$ can be mapped to a circle in $\mathcal{C'}$ of non-smaller radius,
thus we can obtain a packing for $\mathcal{C}$ with the following steps: pack
each circle in the first group in a new bin; for every other circle, pack at the
position of the mapped circle in $P'$. Thus, we have obtained a packing of
$\mathcal{C}$ that uses at most $\OPT_{w\times h}(\mathcal{C'}) + Q \le
\OPT_{w\times h}(\mathcal{C}) + \varepsilon n \Area(\delta) \le
(1+\varepsilon)\OPT_{w\times h}(\mathcal{C})$ bins.

If $n \le K$, then the number of different radii in $\mathcal{C}$ is at most
$K$, otherwise, if $n > K$, the number of different radii in $C'$ is at most
$\lceil \frac{n}{Q} \rceil = \lceil \frac{n}{\lfloor \varepsilon n \Area(\delta)
\rfloor} \rceil \le \lceil \frac{2n}{\varepsilon n \Area(\delta)} \rceil = K$.
In either case, using Lemma~\ref{lemma:fixed_radii}, we can conclude that the
algorithm is polynomial.
\myqed\end{proof}

\section{An Asymptotic PTAS for Circle Bin Packing}\label{sec:aptas}

In this section, we consider the bin packing problem of circles of any size. The
main idea works as follows. First, we find a subset of circles with
``intermediate'' sizes, such that, after the removal of this set, we can divide
the remaining circles into a sequence of groups where the smallest circle in a
group $i$ is much larger than the largest circle in the next group $i+1$, and so
on. The set of removed circles is chosen so that its area is small (\emph{i.e.},
it is bounded by an $\varepsilon$ fraction of the overall circles' area), and
thus we can pack them into separate bins using any constant-ratio approximation
algorithm. To pack the sequence of groups, we do the following: we pack the
first group (of ``large'' circles) using the algorithm from
Section~\ref{sec:large}, and obtain a packing into bins of the original size;
then, we consider sub-bins with a small fraction of the original size, and solve
the problem of packing the remaining groups (of ``small'' circles) in such
sub-bins recursively. To obtain a solution of the original problem, we place
each obtained small bin into the free space of the packing obtained for large
circles.

The key idea to obtain an APTAS is that, if the size of the small bins is much
smaller than the size of large circles, then the waste of space in the packing
of  the large circles is proportional to a fraction of large circles' area.
Moreover, if the size of a such small bin is also much larger than the size of
small circles, then restricting the packing of small circles to small bins does
not increase much the cost of a solution.

\subsection{The algorithm}

In the following, if $B$ is a circle or rectangle, then we denote by $\Area(B)$
the area of $B$. Also, if~$D$ is a set, then $\Area(D) = \sum_{B \in D} \Area(B)$.
We give first a formal description in Algorithm~1; an informal description is
given thereafter.

\bigskip\noindent{}\textbf{Algorithm 1} \emph{Circle bin packing algorithm}

Consider the parameters $r$ and $\gamma$, such that $r$ is a positive integer
multiple of $3$, and $\gamma > 0$. The algorithm receives an instance
$(\mathcal{C}, w, h)$ for the circle bin packing, and returns a packing
of~$\mathcal{C}$ into a set of bins of size~$w \times  (1+\gamma)h$. It is
assumed that $w \le h$, and $hr/w$ is an integer.

\begin{enumerate}\itemsep=0pt
  \item\label{item:eps} Let $\varepsilon = 1/r$;

  \item\label{item:setsG} For every integer $i \ge 0$, define $G_i = \{j \in
  \mathcal{C} : \varepsilon^{2i}w \ge 2 r_j > \varepsilon^{2(i+1)}w \}$;

  \item\label{item:setsH} For each $0 \le j < r$, define $H_j = \{ \ell \in
  G_i : i \equiv j \pmod{r}\}$;

  \item\label{item:ht} Find an integer $t$ such that $\Area(H_t) \le \varepsilon
  \Area(\mathcal{C})$;

  \item\label{item:nfdh} Place each circle of $H_t$ into its bounding box,
  and pack the boxes in separate bins of size \mbox{$w \times (1+\gamma)h$} using NFDH
  strategy~\cite{MeirM68};

  \item\label{item:setsS} For every integer $j \ge 0$, define $ S_j =
  \bigcup_{i\; = \; t+(j-1)r + 1}^{t + jr - 1} G_i $; \quad
    (\emph{see Figure~\ref{fig:partition}})

  \item\label{item:sizes} Define $w_0 = w$, $h_0 = h$, and $w_j = h_j =
  \varepsilon^{2(t+(j-1)r)+1} w$ for every $j \ge 1$;

  \item\label{item:fzero} Let $F_0 = \emptyset$;

  \item\label{item:stepj} For every $j \ge 0$:

  \begin{enumerate}
    \item\label{item:stepj-teo} Use the algorithm of
    Theorem~\ref{theorem:packing_large} to obtain a packing of circles $S_j$
    into bins of size~$w_j \times (1+\gamma)h_j$. Let $P_j$ be the set of
    such bins;

    \item Let $A_j$ be a set of $\max\{|P_j| - |F_j|, 0\}$ new empty bins of
    size~$w_j \times (1+\gamma)h_j$;

    \item Place each bin of $P_j$ into one distinct bin of $F_j \cup A_j$;

    \item Set $F_{j+1} = \emptyset$, and $U_j = \emptyset$; \quad
    (\emph{$U_j$ is used only in the analysis})

    \item For each bin $B$ of $F_j \cup A_j$:

    \begin{itemize}
      \item Let $V$ be the set of bins corresponding to the cells of the
      grid with cells of size~$w_{j+1} \times (1+\gamma)h_{j+1}$ over $B$;

      \item Add to $F_{j+1}$ all bins in $V$ that do not intersect any
      circle of $S_j$.

      \item Add to $U_j$ all bins in $V$ that intersect a circle of $S_j$.
    \end{itemize}

    \item If all circles are packed, go to step~\ref{item:fim-alg}.
  \end{enumerate}

  \item\label{item:fim-alg} Place the bins $A_0, A_1, \dots$ into the minimum
  number of bins of size~$w \times (1+\gamma)h$.
\end{enumerate}

\begin{figure}
\begin{equation*}
\def\arraystretch{1.2}\arraycolsep=2pt
\begin{array}{|rrrcl|c|}
\hline
                 &                &             &          &             & H_t\!:\\
S_0\mathclose{:} &                & \;\;\; G_0, &   \dots  & G_{t-1}     & G_t\\
S_1\mathclose{:} & G_{t+1},       & G_{t+2},    &   \dots  & G_{t+r-1}   & G_{t+r}\\
S_2\mathclose{:} & G_{t+r+1},     & G_{t+r+2},  &   \dots  & G_{t+2r-1}  & G_{t+2r}\\
S_j\mathclose{:} & G_{t+(j-1)r+1},& G_{t+jr+2}, &   \dots  & G_{t+jr-1}  & G_{t+jr}\\
                 &                & \vdots      &          &             & \;\;\vdots \\
\end{array}
\end{equation*}
\caption{Partitioning of the set of circles.\label{fig:partition}}
\end{figure}

In step~\ref{item:setsG}, we group the set of circles into disjoint subsets,
$G_i$, with exponentially decreasing radii. Then, in step~\ref{item:setsH},
these groups are joined into disjoint bunches, $H_i$, according to the remainder
of the division of the group index by $r$. This induces the creation of $r$
bunches, so in step~\ref{item:ht} we may find a bunch $H_t$ whose area
corresponds to a fraction of at most $1/r = \varepsilon$ of the overall circles'
area. In step~\ref{item:nfdh}, we pack the circles of $H_t$ using the NFDH
strategy, that is, we pack the circles $H_t$ greedily in ``shelves'' and in
decreasing order of radii, as in Figure~\ref{fig:nfdh}.

\begin{figure}[b]
\centering
\begin{tikzpicture}[x= 130pt, y=130pt]
\newcommand{\squarewithcircle}[3]{
  \draw[fill=black!10] (#1 + #3 / 2.0, #2 + #3 / 2.0) circle (#3 / 2.0);
  \draw (#1,#2) rectangle (#1 + #3, #2 + #3);
}
\draw[thick] (0, 0.8) -- (0,0) -- (1,0) -- (1,0.8);

\pgfmathparse{0}
\global\let\x\pgfmathresult
\pgfmathsetmacro{\y}{0}

\foreach \s in {0.35, 0.32, 0.28}{
  \squarewithcircle{\x}{\y}{\s}
  \pgfmathparse{\x +\s}
  \global\let\x\pgfmathresult
}

\pgfmathparse{0}
\global\let\x\pgfmathresult
\pgfmathsetmacro{\y}{0.35}
\foreach \s in {0.22, 0.20, 0.18, 0.16, 0.14}{
  \squarewithcircle{\x}{\y}{\s}
  \pgfmathparse{\x +\s}
  \global\let\x\pgfmathresult
}

\pgfmathparse{0}
\global\let\x\pgfmathresult
\pgfmathsetmacro{\y}{0.57}
\foreach \s in {0.12, 0.12, 0.11, 0.11, 0.10, 0.09, 0.09, 0.08, 0.07, 0.06}{
  \squarewithcircle{\x}{\y}{\s}
  \pgfmathparse{\x +\s}
  \global\let\x\pgfmathresult
}
\end{tikzpicture}
\begin{minipage}[t]{0.8\linewidth}
We pack each circle in the corresponding bounding box, and then use the NFDH
strategy for rectangle bin packing.
\end{minipage}
\caption{Packing medium circles of $H_t$.\label{fig:nfdh}}
\end{figure}

In step~\ref{item:setsS}, the remaining unpacked circles are joined in maximal
sequences, $S_j$, of consecutive groups. This leads to a partition of
$\mathcal{C}$ into sets $H_t, S_0, S_1, \dots$ The considered groups and sets
can be distributed over a table, as depicted in Figure~\ref{fig:partition},
where each row corresponds to a sequence~$S_j$, and the last column corresponds
to bunch $H_t$.

For each~$j$, we consider the subproblem of packing the circles of $S_j$ into
bins of size $w_j \times h_j$, as defined in step~\ref{item:sizes}.
This definition guarantees that the circles in $S_j$ are large when compared to
bins of size $w_j \times h_j$, and so we can use the algorithm for large circles
of Theorem~\ref{theorem:packing_large}. Also, the bins considered in the next
iteration have size $w_{j+1}\times h_{j+1}$, and are much smaller than the
circles of $S_j$. Finally, since the circles in $H_t$ are considered separately,
each remaining unpacked circle (of $S_{j+1}, S_{j+2}, \dots$) fits in a bin of
size $w_{j+1}\times h_{j+1}$.

In step~\ref{item:stepj}, we solve each subproblem iteratively. In each
iteration $j\ge 0$, the algorithm keeps a set $F_j$ of free bins of size $w_j
\times (1+\gamma)h_j$ obtained from previous iterations. We obtain a packing of
circles of $S_j$ into a set of bins $P_j$. Then, we place such bins into the
free space of $F_j$, or into additional bins $A_j$, if necessary. The set of
sub-bins of $F_j \cup A_j$ of size $w_{j+1} \times (1+\gamma)h_{j+1}$ that
intersect circles of $S_j$ are included in the set of used sub-bins $U_j$, and
the remaining free sub-bins are saved in $F_{j+1}$ for the next iteration. The
algorithm finishes when all circles are considered, and the created bins $A_0,
A_1, ...$ are combined into bins of size $w \times (1+\gamma) h$ in
step~\ref{item:fim-alg}.

\medskip

We remark that the assumption that $hr/w$ is integer is without loss of
generality. If this is not the case, then, after the first iteration of the
algorithm, we could extend the height $h_0$ of bins in $P_0$ to the next integer
multiple of $w/r$. The final solution would have the property that circles
packed in the extended area are completely packed in the top cells of size $w_1
\times h_1$. Thus, one can easily modify such a solution by moving these cells
to at most $\lceil 1/r\, |P_0| \rceil \le O(\varepsilon)\OPT(\mathcal{C}) + 1$ new
bins of size $w \times h$. For the sake of clarity, from now on, we assume that
$hr/w$ is integer.

\subsection{Analysis}

Now, we analyze the algorithm. The idea of the analysis is that the generated
solution is almost optimal, except that it is restricted in certain ways. The
main strategy will be modifying an optimal solution, so that it respects the
same kind of constraints of a solution given by the algorithm. Also, the cost
increase due to these modifications will be bounded by a small fraction of the
optimal cost. In what follows, we assume that we have run Algorithm~1, giving as
input an instance for the circle bin packing $(\mathcal{C}, w, h)$, and
parameters $r\in \Z_+$, $\gamma > 0$.

\medskip

First, we give some definitions that will be useful for the analysis. In each
iteration, the algorithm considers small sub-bins, that are obtained by creating
grids of elements with size~${w_j \times h_j}$, for some $j \ge 0$, over the
larger bins of size $w \times h$. This notion is formalized in the following.

\begin{definition}
Consider a bin $B$ of size~$w_B \times h_B$. We say that $B$ \emph{respects} $w
\times h$ if $w_B = w_j$, and $h_B = h_j$ for some $j \ge 0$.
Also, if $D$ is a set of bins, then we say that $D$ respects $w \times h$ if
every $B \in D$ respects $w \times h$.
\end{definition}

Next, the function $\Gr_j(B)$ will denote the set of elements in the grid of
sub-bins of size $w_j \times h_j$ over one bin $B$.

\begin{definition}
Let $j \ge 0$. For a given bin $B$, we denote by $\Gr_j(B)$ the set of elements
of the grid of sub-bins of size~$w_j \times h_j$ over $B$. Also, if $D$ is a set
of bins, then $\Gr_j(D) = \cup_{B \in D} \Gr_j(B)$.
\end{definition}

We also need a definition to account the estimate number of bins of original
size $w \times h$, that are needed to pack a set of sub-bins or circles.

\begin{definition}
If $B$ is a rectangle or circle, then $\Ar(B) = \Area(B) /(wh)$. Also, if $D$
is a set of rectangles or circles, then $\Ar(D) = \sum_{B \in D} \Ar(B)$.
\end{definition}

We remark that, if a set of bins $D$ respects $w \times h$, then bins of $D$ can
be easily combined into bins of size~$w \times h$ using almost the same area.
Thus, $\Ar(D)$ is an estimate on the number of bins of size $w \times h$ needed
to pack $D$. Indeed, if there is a packing of a set of circles $C$ into a set of
bins $D$ that respects $w \times h$, then there is a packing of $C$ in
$\lceil \Ar(D) \rceil$ bins of size~$w \times h$.

\smallskip

Our algorithm also deals with bins of size $w_j \times (1+\gamma) h_j$. Each of
the definitions above has a counterpart for the augmented bins. We say that a
bin $B$ of size $w_B \times h_B$ respects $w\times (1+\gamma)h$ if $w_B = w_j$
and $h_B = (1+\gamma)h_j$ for some $j\ge 0$. Also, we decorate with a tilde the
analogous versions of $\Gr$ and $\Ar$, so $\TGr_j(B)$  represents the grid of
sub-bins of size~$w_j \times (1+\gamma)h_j$ over a bin $B$, and $\TAr(B) =
\Area(B) /((1+\gamma)wh)$.

\subsubsection{Bounding the wasted area}

The solutions obtained by Algorithm~1 may be suboptimal mostly because of the
unused area. For instance, for some $j$, the area of bins in $U_j$ is not fully
used, since there might be elements of $U_j$ that intersect circles of $S_j$
only partially. This waste is bounded by the following lemmas.

\begin{lemma}\label{lemma:area}
Let $C \in S_j$ be a circle packed into a bin $B$, and let $D \subseteq
\TGr_{j+1}(B)$ be the subset of bins in the grid that intersect circle $C$, but
are not contained in~$C$. Then $\TAr(D) \le 16 \varepsilon \Ar(C)$.
\end{lemma}

\begin{proof}
Let $r_c$ be the radius of $C$. Each element of $D$ has width~$w_{j+1} =
\varepsilon^{2(t+jr)+1} w$, and height $(1+\gamma)h_{j+1} =
(1+\gamma)\varepsilon^{2(t+jr)+1} w$. Also, since $C \in S_j$, we have $2 r_c
\ge \varepsilon^{2(t+jr)}w$.
Consider the circles $C_+$ and $C_-$, centered at the same point as $C$, and
with radii $r_+ = r_c + w_{j+1} +  (1+\gamma)h_{j+1}$ and $r_- = r_c - w_{j+1} -
(1+\gamma)h_{j+1}$. Notice that every element of $D$ is contained in $C_+\setminus
C_-$.
We obtain
\begin{align*}
\Area(D) &\le \Area(C_+) - \Area(C_-) =  \pi (r_+^2 - r_-^2)\\
        &\le \left((1 + 2\varepsilon + (1+\gamma)2\varepsilon)^2  - (1 - 2\varepsilon - (1+\gamma)2\varepsilon)^2\right)\pi r_c^2\\
  &\le (1+\gamma)16\varepsilon \Area(C).
\end{align*}
Therefore, $\TAr(D) = \Area(D)/((1+\gamma)wh) \le 16\varepsilon \Area(C) / (wh) =
16\varepsilon\Ar(C)$.
\myqed\end{proof}

The next lemma is obtained analogously.

\begin{lemma}\label{lemma:area2}
Let $C \in S_j$ be a circle packed into a bin $B$, and let $D \subseteq
\Gr_{j+1}(B)$ be the subset of bins in the grid that intersect circle $C$, but
are not contained in $C$. Then $\Ar(D) \le 16\varepsilon \Ar(C)$.
\end{lemma}

\subsubsection{Modifying an optimal solution}\label{subsub-mod}

In the following, we will show that requiring that each set of circles $S_j$ be
packed into grid bins of size $w_j \times h_j$ does not increase much the
solution cost. This fact is central to the algorithm, since it allows packing
sets $S_j$'s iteratively.
To show these properties, we will transform an optimal packing $Opt$ of
$\mathcal{C}$ into a packing $D$ with the desired properties. The idea is moving
circles of $S_j$ that intersect lines of the grid of size $w_j \times h_j$ to
free bins that respect the grid.

Notice that, in the optimal solution $Opt$, the circles of $S_0$ already respect
the grid of size~${w_0 \times h_0}$, since $w = w_0$ and $h = h_0$. For each $j
\ge 1$, we know that the diameter of a circle in $S_j$ is at most $\varepsilon
w_j$, so the total area of circles in $S_j$ that intersect the grid lines
is~$O(\varepsilon) \Area(Opt)$. To fix the intersections, we may consider moving
such circles to new bins of size $w_j \times h_j$ with area at most
$O(\varepsilon) \Area(Opt)$. As a first attempt, one could repeat this process
for each set $S_j$, however, that would increase the solution cost by
$O(\varepsilon) \Area(Opt)$ for each $j\ge 1$. Therefore, we need a more involved
approach: first, we move intersecting circles of $S_1$ to newly created sub-bins
of size $w_1\times h_1$; then, for each $j > 1$, we move the intersecting
circles of $S_j$ to sub-bins of size $w_j\times h_j$ placed over the space left
by circles of $S_{j-1}$ that intersected the grid lines, and that were moved in
the previous iteration.

The next algorithm keeps the invariant that, at the start of iteration $j\ge1$,
the set $R_j$ contains free space to pack all the grid intersecting circles of $S_j$.
Steps~(\ref{item:movebegin})-(\ref{item:group}) move intersecting circles to
bins of $R_j$, and steps~(\ref{item:spacebegin})-(\ref{item:spaceend}) prepare
for the next iteration by making sure that there are enough free bins in
$R_{j+1}$ respecting the grid of size $w_{j+1} \times h_{j+1}$.

For $j\ge1$, we will fix the intersections of $S_j$ by considering each bin $B$
of the grid $\Gr_j(Opt)$ iteratively, in step~(\ref{item:eachbin}). We will
consider 4 very thin rectangles over the boundary of $B$, as in
Figure~\ref{fig:lines}. The set of all such rectangles will be $L_j$. Here, we
take advantage of the fact that the circles are non-overlapping, and move each
rectangle containing the intersecting circles to the next free space of $R_j$,
preserving the packing arrangement (alternatively, one could repack all
intersecting circles using a different constant approximation algorithm).

To reserve free space in $R_{j+1}$ for the next iteration, we will split the
rectangles in $L_j$ into sub-bins of size $w_{j+1} \times h_{j+1}$. We use as
much as possible the free sub-bins of $L_j$. If a given sub-bin $B$ of $L_j$ is
used by a circle, then we may use the sub-bin of $Opt$, denoted by $\phi(B)$,
that was originally occupied by the corresponding region of this circle. There
might be cases that the sub-bin $B$ intersects the circle only partially, so
that $\phi(B)$ is potentially used by other circles. In such cases, we create a
new bin of size $w_{j+1} \times h_{j+1}$.

The complete algorithm to modify an optimal solution is formalized in the
following.

\begin{figure}
\centering
\includegraphics[width=4.7cm]{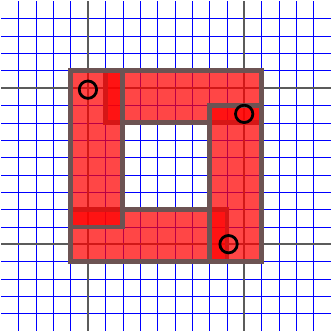}

\begin{minipage}[t]{0.8\linewidth}
Each circle of $W$ has diameter at most $\varepsilon w_j$. To  rearrange the
rectangles into bins of size $w_j \times w_j$ (of $R_j$), we use one side of
length $w_j$. To ensure every circle is in a rectangle, the other side has
length $3 \varepsilon w_j$ (see lower circle).
\end{minipage}
\caption{Removing circles that intersect grid lines.\label{fig:lines}}
\end{figure}

\begin{enumerate}
 \item Let $R_1$ be a set of $12 \varepsilon (wh)/(w_1h_1) |Opt|$ new bins of
 size~$w_1 \times h_1$;

 \item Let $D_0 = Opt \cup R_1$;

 \item For each $j \ge 1$:

 \begin{enumerate}

  \item[] \emph{Move to $R_j$ the circles in $S_j$ that intersect grid lines:}

  \item\label{item:movebegin} Let $L_j = \emptyset$;

  \item\label{item:eachbin} For each bin $B \in \Gr_j(Opt)$:

  \begin{enumerate}
    \item Let $W$ be the set of circles in $S_j$ that intersect the boundary
    of $B$;

    \item Let $V$ be a set of 4 new bins (2 of size~$3 \varepsilon w_j
    \times h_j$, and 2 of size~$w_j \times 3 \varepsilon h_j$) placed over
    the boundary of $B$, so that each circle in $W$ is contained in one bin
    of $V$ (see Figure~\ref{fig:lines});

    \item\label{item:defphi} For each cell $B' \in \Gr_{j+1}(V)$, let
    $\phi(B')$ be the cell of $\Gr_{j+1}(Opt)$ under~$B'$;

    \item Remove each circle of $W$ from the packing $D_j$ and pack it over
    one bin of $V$ preserving the arrangement;

    \item Add $V$ to $L_j$;
  \end{enumerate}

  \item\label{item:group} Make groups of $r/3$ bins (of equal sizes) of $L_j$
  forming new bins of size~$w_j \times h_j$, and place each such bin over
  one distinct element of $R_j$;

  \medskip

  \item[] \emph{Reserve space for $R_{j+1}$:}

  \item\label{item:spacebegin} Let $R_{j+1} = \emptyset$;

  \item Let $N_j = \emptyset$;

  \item\label{item:spaceend} For each bin $B \in \Gr_{j+1}(L_j)$, consider the
  cases:

    \begin{enumerate}
      \item\label{item:case1} If $B$ does not intersect any circle of
      $S_j$, then add $B$ to $R_{j+1}$;

      \item\label{item:case2} If $B$ is contained in some circle of $S_j$,
      then add $\phi(B)$ to $R_{j+1}$;

      \item\label{item:case3} If $B$ intersects, but is not contained in a
      circle of $S_j$, then create a new bin of size~$w_{j+1} \times
      h_{j+1}$ and add it to $R_{j+1}$, and to $N_j$;
    \end{enumerate}

  \item Let $D_j = D_{j-1} \cup N_j$;

  \item If $\mathcal{C}\setminus H_t = S_0 \cup \dots \cup S_j$, make $D =
  D_j$, and stop.
 \end{enumerate}
\end{enumerate}

We will show that $D$ is a bin packing of $\mathcal{C}\setminus H_t$. First, we note that the
procedure is well defined. It is enough to check that $R_j$ has free space to
pack bins of $L_j$. Indeed, it is not hard to show the following claim.

\begin{claim}\label{claim-space}
For every $j \ge 1$, $\Area(L_j) = \Area(R_j) = \Area(R_1)$.
\end{claim}
\begin{proof}
First, note that $\Area(L_j) = \Area(R_1)$ for every $j\ge1$. Indeed, we have
$|\Gr_j(Opt)| = |Opt| (wh)/(w_jh_j)$, and for each bin $B$ in
$\Gr_j(Opt)$, we create a set~$V$ such that ${\Area(V) = 4\cdot3\varepsilon
w_jh_j}$, so $\Area(L_j) = 12\varepsilon (wh)|Opt| = \Area(R_1)$.

Also, note that $\Area(R_{j+1}) = \Area(L_j)$, since for each $B \in
\Gr_{j+1}(L_j)$, $\Area(R_{j+1})$ is increased by $\Area(B)$. This is clear if
steps~(\ref{item:case1}) and~(\ref{item:case3}) are executed, so it is enough to
show that, if step~(\ref{item:case2}) is executed for bins $B, B'$ in
$\Gr_{j+1}(L_j)$, with $\phi(B) = \phi(B')$, then~${B = B'}$. By the definition
in step~(\ref{item:defphi}), $B$ and $B'$ must intersect the same region of a
circle in $S_j$. Since each such circle is contained in exactly one rectangle of
$L_j$, it follows that, indeed,~${B = B'}$. Then $\Area(R_{j+1}) = \Area(L_j)$.

Therefore $\Area(R_j) = \Area(R_1)$ for every $j\ge1$.
\myqed\end{proof}

\smallskip

Now, it will be shown that the algorithm produces a modified solution with the
desired properties.

\begin{claim}
At the end of iteration $j \ge 0$, the following statements hold:

\begin{enumerate}
 \item $D_j$ is a packing of $\mathcal{C}$;

 \item for each $\ell = 0,  \dots, j$, there is a packing of $S_\ell$ into a
 set $P'_\ell \subseteq \Gr_\ell(D_j)$ of bins of size~$w_\ell \times h_\ell$;

 \item the bins in $R_{j+1} \subseteq \Gr_{j+1}(D_j)$ do not intersect any
 circle of $\mathcal{C}$.
\end{enumerate}
\end{claim}

\begin{proof}
By induction on $j$. For $j = 0$, the statements are clear. So let $j \ge 1$,
and assume that the statements are true for $j - 1$.

\textit{Statement 1: } Clearly, $L_j$ is a packing of the circles that were
removed from the original packing $D_{j-1}$. Since $r$ is a multiple of $3$ and
by Claim~\ref{claim-space}, step~(\ref{item:group}) is well defined, and we can
place rectangles of $L_j$ over bins of $R_j$. After step~(\ref{item:group}), we
have a bin packing of $\mathcal{C}$, since, by the induction hypothesis, the set
$R_j$ did not intersect any circle at the beginning of iteration $j$. This shows
statement~1.

\textit{Statement  2: } Since, at the end of iteration, each circle of $S_j$
that intersected a line of the grid $\Gr_j(D_{j-1})$ is completely contained in
a bin of $R_j \subseteq \Gr_j(D_j)$, we obtain statement~2.

\textit{Statement   3: } If step~(\ref{item:case1}) or step~(\ref{item:case3})
is executed, then we add a free bin to $R_{j+1}$. Thus, we only need to argue
that, whenever step~(\ref{item:case2}) is executed, the bin $\phi(B)$ does not
intersect any circle. Let $C$ be the circle that contains $B$, so that at the
beginning of the iteration,~$\phi(B)$ was contained in $C$. Since $C$ was moved
to $L_j$, $\phi(B)$ does not intersect any circle when step~(\ref{item:case2})
is executed. This completes the proof.
\myqed\end{proof}

Finally, we may calculate the cost of the modified solution $D$.

\begin{claim}
 $\Ar(D) \le (1 + 28\varepsilon)\OPT_{w\times h}(\mathcal{C})$.
\end{claim}
\begin{proof}
Let $m$ be the number of iterations of the algorithm that modifies $Opt$. Notice
that $D$ is the disjoint union of $Opt$, $R_1$, $N_1$, \dots, $N_m$. Thus we get
\begin{align*}
 \Ar(D) &=   \textstyle \Ar(Opt) + \Ar(R_1) + \sum_{j = 1}^m \Ar(N_j) \\
         &\le \textstyle \Ar(Opt)+12\varepsilon (wh)|Opt| / (wh) +\sum_{j=1}^m 16\varepsilon\Ar(S_j)\\
         &\le \Ar(Opt)+12\varepsilon \Ar(Opt) +16\varepsilon \Ar(\mathcal{C})\\
         &\le (1 + 28\varepsilon)\OPT_{w\times h}(\mathcal{C}),
\end{align*}
where the first inequality comes from Lemma~\ref{lemma:area2}.
\myqed\end{proof}

By combining the last claims, we obtain the following lemma.

\begin{lemma}\label{lemma:modopt}
There is a packing of $\mathcal{C} \setminus H_t$ into a set of bins $D$ that
respects $w \times h$ with~${\Ar(D) \le (1+28\varepsilon)\OPT_{w\times
h}(\mathcal{C})}$, such that for every $j \ge 0$, there is a packing of $S_j$
into a set of bins $P'_j \subseteq \Gr_j(D)$.
\end{lemma}

In addition to requiring that each set of circles $S_j$ be packed into bins of
the grid with cells of size~$w_j \times h_j$, we also require that the bins used
to pack $S_{j+1},S_{j+2}, \dots$ do not intersect circles of $S_1, \dots,
S_{j-1}$. Again, this restriction only increases the cost of a solution by a
small fraction of the optimal value, as shown by the next lemma.

\begin{lemma}\label{lemma:modopt2}
There is a packing of $\mathcal{C} \setminus H_t$ into a set of bins $D$ that
respects $w \times h$ with $\Ar(D) \le (1+44\varepsilon)\OPT_{w\times
h}(\mathcal{C})$, such that for every $j \ge 0$, there is a packing of $S_j$
into a set of bins $P'_j \subseteq \Gr_j(D)$. Moreover, if $B \in P'_j$, then
$B$ does not intersect any circle in $S_\ell$ for $\ell < j$.
\end{lemma}

\begin{proof}
We start with a solution $D'$ given by Lemma~\ref{lemma:modopt}, and
corresponding bin packings $P'_j$ for each $j$. Without loss of generality, we
assume that for each $B \in P'_j$, there is a circle $i \in S_j$ packed in $B$,
since otherwise we could simply remove $B$ from $P'_j$. We execute the following
steps:
\begin{enumerate}
  \item Let $P''_0 = P'_0$, and $D = D'$;

  \item For each $j\ge 1$:

  \begin{enumerate}

    \item Let $V \subseteq \Gr_{j}(D)$ be the set of circles that intersect
    a circle $C \in S_{j-1}$, but are not contained in $C$.

    \item Create a set $V'$ of $|V|$ new bins of size~$w_j \times h_j$, and
    move the circles in $V$ to $V'$ preserving the arrangement.

    \item Let $P''_j = P'_j \cup V'$, and add $V'$ to $D$.

    \item If $\mathcal{C}\setminus H_t = S_0 \cup \dots \cup S_j$, stop.
  \end{enumerate}
\end{enumerate}

Clearly, the output of this procedure is a bin packing of $\mathcal{C}$, and a
simple induction shows that at the end of iteration $m$, for every $j \le m$,
set $P''_j \subseteq \Gr_j(D)$ is a bin packing of $S_j$, and for every $B \in
P'_j$, $B$ does not intersect any circle in $S_\ell$ for $\ell < j$. Using
Lemma~\ref{lemma:area2}, we obtain
\begin{align*}
\Ar(D)  &\le \textstyle \Ar(D') + \sum_{j\ge0} 16 \varepsilon\Ar(S_j)\\
         &\le (1 + 28\varepsilon)\OPT(\mathcal{C}) + 16 \varepsilon \OPT(\mathcal{C})\\
         &= (1 + 44\varepsilon)\OPT(\mathcal{C}). \qedhere
\end{align*}
\end{proof}

\subsubsection{Obtaining an APTAS}

Now we calculate the cost of the solution generated by Algorithm~1.

\begin{lemma}\label{lemma:aprox}
Algorithm~1 produces a packing of $\mathcal{C}$ into bins of size~$w \times
(1+\gamma)h$ using at most $(1+O(\varepsilon))\OPT_{w\times h}(\mathcal{C}) + 2$
bins.
\end{lemma}
\begin{proof}
Let $D$ be the packing of $\mathcal{C}$ obtained from Lemma~\ref{lemma:modopt2},
and let sets $P'_j$, $j \ge 0$, be the packings obtained for each $S_j$. Notice
that for each $j\ge 0$, $\OPT_{w_j\times h_j}(S_j) \le  |P'_j|$. Recall that
Algorithm~1 uses Theorem~\ref{theorem:packing_large} to obtain a packing $P_j$
of $S_j$ with $|P_j| \le (1+\varepsilon)\OPT_{w_j\times h_j}(S_j)$ bins.
Therefore, $|P_j| \le (1+\varepsilon)|P'_j|$, and thus $\TAr(P_j)\le
(1+\varepsilon)\Ar(P'_j)$.

First, we show, by induction on $m$, that for every $m \ge 0$, we have
$\TAr(F_m) = \sum_{j = 0}^{m-1} \TAr(A_j) - \sum_{j = 0}^{m-1} \TAr(U_j)$. For
$m = 0$, the claim is clear, so suppose that the claim is true for some $m\ge
0$. We have,
\begin{align*}
\TAr(F_{m+1}) &= \TAr(F_m \cup A_m) - \TAr(U_m) \\
             &= \TAr(F_m) + \TAr(A_m) - \TAr(U_m) \\
             &= \textstyle\sum_{j = 0}^{m-1} \TAr(A_j) - \sum_{j = 0}^{m-1} \TAr(U_j) + \TAr(A_m) - \TAr(U_m) \\
             &= \textstyle\sum_{j = 0}^m \TAr(A_j) - \sum_{j = 0}^{m} \TAr(U_j).
\end{align*}

Now, we show that  for every $m \ge 0$, we have $\sum_{j=0}^m \TAr(A_j) \le
(1+\varepsilon)\Ar(D) + 15\varepsilon \sum_{j=0}^m \Ar(S_j)$. Again, by
induction on $m$. The case $m = 0$ is clear, so suppose that this is true for
some $m\ge 0$. We consider two cases. If $|A_{m+1}| = 0$, then by the induction
hypothesis we have
\begin{align*}
\textstyle \sum_{j=0}^{m+1} \TAr(A_j) &=   \textstyle
                             \sum_{j=0}^m \TAr(A_j)\\
                          &\le  \textstyle
                             (1+\varepsilon)\Ar(D) + 15\varepsilon \sum_{j=0}^m \Ar(S_j)\\
                          &\le  \textstyle
                             (1+\varepsilon)\Ar(D) + 15\varepsilon \sum_{j=0}^{m+1} \Ar(S_j).
\end{align*}
If $|A_{m+1}| > 0$, then we know that
that $\TAr(A_{m+1}) +\TAr(F_{m+1}) = \TAr(P_{m+1})$, so
we get
\begin{align*}
\sum_{j=0}^{m+1} \TAr(A_j) &= \sum_{j=0}^m \TAr(A_j) + \TAr(A_{m+1}) \\
                          &= \left(\sum_{j = 0}^{m} \TAr(U_j) + \TAr(F_{m+1})\right)
                            + \left( \TAr(P_{m+1}) - \TAr(F_{m+1}) \right) \\
                          &\le (1+ 16\varepsilon) \sum_{j = 0}^m\Ar(S_j) + (1+\varepsilon)\Ar(P'_{m+1})\\
                          &= (1+\varepsilon)\left( \sum_{j = 0}^m\Ar(S_j)+ \Ar(P'_{m+1})\right)
                          + 15 \varepsilon \sum_{j = 0}^{m}\Ar(S_j)\\
                          &\le (1+\varepsilon)\Ar(D)
                          + 15 \varepsilon \sum_{j = 0}^{m+1}\Ar(S_j).
\end{align*}
The first inequality comes from Lemma~\ref{lemma:area},
and the second inequality comes from the fact that bins in $P'_{m+1}$ do not
intersect circles in $S_1, \dots\, S_m$.

It follows that the total area of bins used for circles in $\mathcal{C}\setminus H_t$ is
\begin{align*}
\textstyle \sum_{j\ge 0} \TAr(A_j)
      &\le \textstyle (1+\varepsilon)\Ar(D) + 15 \varepsilon \sum_{j\ge0}\Ar(S_j)\\
      &\le (1+\varepsilon)(1+44\varepsilon)\OPT(\mathcal{C}) + 15 \varepsilon \OPT(\mathcal{C})\\
      &\le (1+105\varepsilon)\OPT(\mathcal{C}).
\end{align*}

For the circles in $H_t$, we used the NFDH algorithm to pack the bounding boxes
of the circles, thus the density of the packing in each used bin of size~$w
\times (1+\gamma)h$, with the exception of the last, is at least $1/4$. Since
each circle occupies a fraction of $\pi/4$ of its bounding box, the total number
of bins used for $H_t$ is bounded by
\[
  \lceil 4\cdot4/\pi \TAr(H_t) \rceil \le \lceil 4\cdot4/\pi \varepsilon \TAr(\mathcal{C}) \rceil
\le \lceil 16/\pi \varepsilon \OPT(\mathcal{C}) \rceil.
\]

Therefore, noticing that the sets of bins $A_0,A_1, \dots$ respects $w \times
(1+\gamma)h$, we obtain that the total number of bins of size~$w \times
(1+\gamma)h$ used to pack $\mathcal{C}$ is
\[
\lceil (1+105\varepsilon)\OPT(\mathcal{C}) \rceil + \lceil 16/\pi \varepsilon \OPT(\mathcal{C}) \rceil
\le (1+111\varepsilon)\OPT(\mathcal{C}) + 2.\qedhere
\]
\end{proof}

Next lemma shows that running time of Algorithm~1 is polynomial if the ratio of
a bin's dimensions is bounded by a constant.

\begin{lemma}\label{lemma:time}
Suppose $h/w = O(1)$. For any given constant $r$, and number $\gamma > 0$, the running
time of Algorithm~1 is polynomial.
\end{lemma}

\begin{proof}
Notice that in an implementation of Algorithm~1, we only run
step~(\ref{item:stepj}) for nonempty sets $S_j$. In each such iteration, we
only need to account for the running time of step~(\ref{item:stepj-teo}), and
the running time to pack the bins of $P_j$ in free space of current bins.

We consider the following equivalent and alternative procedure for
step~(\ref{item:stepj-teo}):
\begin{enumerate}
  \item scale the radius of each circle in $S_j$ by $1/w_j$ and obtain a
  set $S'_j$;

  \item run algorithm of Theorem~\ref{theorem:packing_large} with the scaled
  $S'_j$ and bins of width~$1$ and height ${h_j/w_j \le h/w \in O(1)}$, and
  obtain a packing $P'_j$;

  \item scale the obtained packing $P'_j$ by $w_j$, and obtain a packing
  $P_j$.
\end{enumerate}

Notice that for every $j$, the radius of the smallest circle in $S'_j$ is at
least a constant ${\delta = \varepsilon^{2(t+jr)}w /(2 w_j) =
\varepsilon^{2(t+jr)}w /(2 \varepsilon^{2(t+(j-1)r)+1} w) = \varepsilon^{2r -
1}}$, thus, Theorem~\ref{theorem:packing_large} implies that the total running
time of step~(\ref{item:stepj-teo}) is polynomial.

To pack a bin of $P_j$,  for $j\ge 0$, we need to find one element $\TGr_j(A_0
\cup \dots \cup A_{j-1})$ that is not in $U_0, \dots, U_{j-1}$, that is, we need
to find a grid cell that does not intersect any circle of $S_0 \cup \dots \cup
S_{j-1}$. To verify whether a grid cell intersects a circle of $S_\ell$, with $0
\le \ell \le j-1$, it is enough to list the elements of $U_\ell$ that intersect
the border of the circle. There is at most a constant number of such elements
per circle, so at most $O(n)$ elements are listed until we find one free cell.
Therefore, Algorithm~1 can be implemented in polynomial time.
\myqed\end{proof}

Combining Lemmas~\ref{lemma:aprox} and~\ref{lemma:time} we obtain our main
theorem.

\begin{theorem}\label{theorem:binpack}
Let $(\mathcal{C}, w, h)$ be an instance of the circle bin packing. For any
constant \mbox{$\varepsilon > 0$}, and number $\gamma > 0$, we can obtain in
polynomial time a packing of $\mathcal{C}$ into at most
$(1+\varepsilon)\OPT_{w\times h}(\mathcal{C})+2$ rectangular bins of size~$w
\times (1+\gamma)h$.
\end{theorem}

\begin{proof}
If $h/w < 1/\varepsilon^2 \in O(1)$, then the theorem is immediate, so assume
$h/w \ge 1/\varepsilon^2$.
Consider an optimal solution $Opt$ of bins of size~$w \times h$. We will
transform this solution into a packing of bins of size~$w \times w/\varepsilon$.
First, split each bin of $Opt$ in sub-bins of size~$w \times w/\varepsilon$.
Then, remove all circles that intersect  consecutive sub-bins. The total area of
removed circles is at most $|Opt| (w\cdot 2w) \lfloor h/(w/\varepsilon) \rfloor
\le |Opt| 2 wh \varepsilon $. Finally, place the removed circles into their
bounding boxes and pack them into additional bins of size~$w \times
w/\varepsilon$ using the NFDH strategy. Since each additional bin has density of
at least $\pi/16$ (with exception of the last), the number of such bins is
bounded by $\lceil (16/\pi\, |Opt| 2 wh \varepsilon)/(w(w/\varepsilon)) \rceil
$. Therefore, we know that $\OPT_{w\times w/\varepsilon}(\mathcal{C})$ is
bounded by $ |Opt|\lceil h/(w/\varepsilon)\rceil + \lceil (16/\pi\, |Opt| 2 wh
\varepsilon)/(w(w/\varepsilon)) \rceil \le (1+O(\varepsilon))|Opt|(h/w) \,
\varepsilon$, where we have used $h/w \ge 1/\varepsilon^2$ in the inequality.

Now, we use Lemma~\ref{lemma:aprox} and obtain a packing of $\mathcal{C}$ into
bins of size~$w \times (1+\gamma)w/\varepsilon$ of cost at most
\[
(1+O(\varepsilon))\OPT_{w\times w/\varepsilon}(\mathcal{C}) + 2
\le
(1+O(\varepsilon))|Opt|(h/w) \, \varepsilon+2.
\]
By joining each group of $\lfloor h /(w/\varepsilon) \rfloor$ bins, we obtain a
packing into bins of size~$w \times (1+\gamma)h$ of cost at most
\[
\left\lceil \frac{(1+O(\varepsilon))|Opt|(h/w) \, \varepsilon+2}{
\lfloor h /(w/\varepsilon) \rfloor}
\right\rceil
\le
(1+O(\varepsilon))|Opt| + 2,
\]
where the inequality follows from the fact that $h/w \ge 1/\varepsilon^2$, and
assuming $\varepsilon$ sufficiently small. To complete the proof, it is enough
to notice that the running time is given by Lemma~\ref{lemma:time}.
\myqed\end{proof}

Now, it is straightforward to extend Theorem~\ref{theorem:binpack} to the circle
strip-packing.

\begin{theorem}\label{theorem:strip}
Let $\mathcal{C}$ be a set of circles. For any given constant $\varepsilon > 0$,
we can obtain in polynomial time a packing of $\mathcal{C}$ in a strip of unit
width and height $(1+\varepsilon)\mathrm{OPTS}(\mathcal{C})+O(1/\varepsilon)$,
where $\mathrm{OPTS}(\mathcal{C})$ is the height of the minimum packing of
$\mathcal{C}$ in a strip of unit width.
\end{theorem}

\section{A Resource Augmentation Scheme for Circle Bin Packing}
\label{sec:resource}

In this section, we discuss the use of augmented bins by our algorithm. The
reason to use bins of height $1+\gamma$ is that the algebraic algorithms used in
Section~\ref{sec:algebraic} only give approximate packing with rational
coordinates, and thus we enlarge the size of the bin to avoid intersections.
Other approaches could be considered to obtain approximate packings, as well.
One alternative would be discretizing the possible locations for circle centers.
The idea is to create a grid of points, for which the distance between two
adjacent locations is $\gamma$, for some small $\gamma > 0$. Then, we could
modify an optimal solution so that the center of each center is moved to the
closest point in the grid, and obtain an $O(\gamma)$-packing. To solve the
problem of packing a constant-sized set of circles in an augmented bin, we could
simply try all possible combinations, and check whether we obtain an
$O(\gamma)$-packing.

Notice that this discretization algorithm does not provide a certificate that
there exists a packing into the original (non-augmented) bin, that is, we may
obtain an approximate packing for a set of circles that cannot be packed in a
bin. When using a quantifier elimination algorithm, we obtain such a
certificate, and the need for approximate packings is due to numeric reasons
only.
Moreover, algebraic algorithms have a much weaker dependency on the size of the
resource augmentation parameter, $\log 1/\gamma$. Namely, algebraic algorithms
have time complexity $O((\log 1/\gamma)^{O(1)})$, while discretization
algorithms would have time complexity $O((1/\gamma)^{O(1)})$. This means that,
if $\gamma$ is part of the input, then the algebraic algorithms would be
polynomial, while the discretization algorithms would not.

Again, the sole reason for having augmented bins in
Theorem~\ref{theorem:binpack} is that we insisted in obtaining a rational
solution. Thus, on the one hand, if we allowed a more general model of
computation, in which one could compute and operate over polynomial solutions,
we would obtain an APTAS for the circle bin packing problem without resource
augmentation. On the other hand, if we admit resource augmentation, and consider
the particular case that a bin is enlarged by a constant size, then we may
strengthen Theorem~\ref{theorem:binpack}, and obtain a packing with no more bins
than in an optimal solution, as in the following result.

\begin{theorem}\label{theorem:binpack_resource}
Let $(\mathcal{C}, w, h)$ be an instance of the circle bin packing. For any
given constant~${\varepsilon > 0}$, we can obtain in polynomial time a
packing of $\mathcal{C}$ into at most $\OPT_{w\times h}(\mathcal{C})$
rectangular bins of size~$w \times (1+\varepsilon)h$.
\end{theorem}

A few sources are responsible for the additional bins used in the algorithm of
Theorem~\ref{theorem:binpack}. In Theorem~\ref{theorem:packing_large}, which was
used for the packing of large circles, we spent additional bins for the first
group of $Q$ circles. In Algorithm~1, the elements in the set of intermediate
circles $H_t$ are packed in separate new bins. Also, we wasted the space of
sub-bins that partially intersected large circles. In the following, we will
discuss how to slightly modify our algorithm, and avoid the use of additional
bins for each case. This will be done by augmenting the bin with small strips of
height $O(\varepsilon)h$, thus obtaining Theorem~\ref{theorem:binpack_resource}.

First we consider an alternative approach for the packing of large circles of
Theorem~\ref{theorem:packing_large}. Instead of mapping large circles to small
circles, we simply round down the radius of each circle, and obtain an
approximate packing.

\begin{lemma}\label{lemma:packing_large_resource}
Let $(\mathcal{C}, w, h)$ be an instance of the circle bin packing, such that
$w, h \in O(1)$, and $\min_{1\le i\le n} r_i \ge \delta$, for some constant
$\delta$.
For any given constant $\varepsilon > 0$, there is a polynomial-time
algorithm that packs $\mathcal{C}$ into at most $\OPT_{w\times h}(\mathcal{C})$
rectangular bins of size~$w \times (1+\varepsilon)h$.
\end{lemma}

\begin{proof}
Let $M = \lceil wh/\Area(\delta) \rceil$  and $\alpha = \varepsilon^2/(6M^2)$. We
obtain a modified instance $\mathcal{C}'$ such that the radius of each circle is
rounded down to a number in the sequence $\delta,\; \delta+ \alpha,\; \delta+
2\alpha, \dots$ Thus, in $\mathcal{C}'$, the number of different radii is at
most a constant. We use Lemma~\ref{lemma:fixed_radii}, and obtain a packing into
at most $\OPT_{w\times h}(\mathcal{C}') \le \OPT_{w\times h}(\mathcal{C})$ bins
of size $w \times (1+\varepsilon) h$. By restoring the original radii, we get a
$2\alpha$-packing of $\mathcal{C}$. Finally, we use
Lemma~\ref{lemma:augmentation}, and obtain a packing of $\mathcal{C}$ into no
more than $\OPT_{w\times h}(\mathcal{C})$ bins of size $w \times
(1+O(\varepsilon)) h$.
\myqed\end{proof}

Now we proceed to prove Theorem~\ref{theorem:binpack_resource}.

{\def\proofname{Proof of Theorem~\ref{theorem:binpack_resource}}
\begin{proof}
First, we pack the circles of $H_t$ into at most $\OPT(\mathcal{C})$ strips of
height~${O(\varepsilon)h}$. We assume, without loss of generality, that $t > 0$,
since otherwise we could modify the algorithm and find some $t' > 0$, such that
$\Area(H_{t'}) \le 2 \varepsilon \Area(\mathcal{C})$. Since $t > 0$, each circle
of $H_t$ has radius at most $\varepsilon h$. Therefore, we can pack $H_t$ into
strips of size $w \times 11 \varepsilon h$ using the NFDH strategy. Using
density arguments, we obtain that the total number of strips is
\begin{align*}
  \lceil 4 \cdot 4/\pi \cdot \Area(H_t) / (11\varepsilon w h) \rceil
  &\le
  \lceil 4 \cdot 4/\pi \cdot 2 \varepsilon \Area(\mathcal{C}) / (11 \varepsilon w h) \rceil\\
  &\le
  \lceil 32 / (11 \pi) \OPT(\mathcal{C}) \rceil
  \le \OPT(\mathcal{C}).
\end{align*}

Now, we only need few more changes to Algorithm~1. First, we set parameter
$\gamma = \varepsilon$. Then, in step~\ref{item:stepj-teo}, we replace
Theorem~\ref{theorem:packing_large} by Lemma~\ref{lemma:packing_large_resource}.
Finally, in step~\ref{item:fim-alg}, rather than packing sub-bins of $A_0, A_1,
...$ into bins of size $w \times (1+\varepsilon) h$, we pack $A_0, A_1, ...$
into bins of size ${w \times (1+105\varepsilon) (1+\varepsilon) h}$. Notice that
the simple first-fit greedy algorithm that packs such sub-bins in decreasing
order of height has the property that, if a new bin is created to pack a
sub-bin~$B$, then either~${B \in A_0}$; or~$B \not\in A_0$ and every other
created bin is fully used. Therefore, it is not hard to see that the total
number of bins is bounded by ${\max\{|A_0|, \; \lceil
\sum_{j\ge0}\Area(A_j)/((1+105\varepsilon) (1+\varepsilon) wh )\rceil \}}$.

By Lemma~\ref{lemma:packing_large_resource}, we may bound the first term as
$|A_0| \le \OPT(S_0) \le \OPT(\mathcal{C})$. To bound the second term,
we may repeat the proof of Lemma~\ref{lemma:aprox}, and
obtain
\[
\textstyle
\sum_{j\ge 0} \TAr(A_j) \le (1+105\varepsilon)\OPT(\mathcal{C}),
\]
which implies
\[
\textstyle
\sum_{j\ge 0} \Area(A_j) \le (1+\varepsilon) (1+105\varepsilon)wh\OPT(\mathcal{C}).
\]
Therefore
\[
 \left\lceil \frac{\sum_{j\ge0}\Area(A_j)}{(1+105\varepsilon) (1+\varepsilon) wh} \right\rceil
 \le
 \left\lceil \frac{(1+\varepsilon) (1+105\varepsilon)wh\OPT(\mathcal{C})}{(1+105\varepsilon) (1+\varepsilon) wh} \right\rceil
 =
 \OPT(\mathcal{C}).
\]
To complete the proof, we combine the bins used for $A_0, A_1, ...$ with the
strips used for $H_t$.
\myqed\end{proof}
}

\section{Generalizations}\label{sec:generalizations}

In this section, we discuss possible generalizations to our algorithm that
preserve the approximation ratio. We notice that, although Algorithm~1 only
considers the case of circles, it can be applied to many other packing problems.
Indeed, Algorithm~1 induces a \emph{unified framework} for packing problems that
satisfy certain assumptions, such as the existence of a polynomial-time
algorithm to approximately pack ``large'' items (in an augmented bin), and bounded
wasted volume of sub-bins that partially intersect an item.

As an illustrative example, in this section, we will consider the problem of
packing $d$-dimensional $L_p$-norm spheres%
\footnote{To comply with the majority of works in the literature, in this paper
we use the term \emph{circle}, rather than disk, to refer to the interior of a
region. Similarly, and for the sake of consistency with the multidimensional
packing literature, we use the term \emph{sphere}, rather than ball, to refer to
the interior of a solid.}
in $d$-dimensional boxes, and give a high-level summary of necessary changes in
the algorithm. Moreover, we will discuss how our algorithm can be used for bins
of different shapes, provided that we can satisfy some conditions.

\subsection{Packing $d$-dimensional $L_p$-norm spheres}

For a vector $v \in \R^d$, we denote by $v(i)$ the $i$-th coordinate of $v$. A
$d$-dimensional box of size $v$ is a hyperrectangle with sides of length $v(i) >
0$, for each~${1 \leq i \leq d}$. In this section, we consider the problem of
packing $L_p$-norm spheres in $d$-dimensional boxes of a given size, for a
rational $p \geq 1$. In fact, we allow a slightly more general concept of norm,
that we call the \emph{weighted $L_p$-norm}, when each dimension can be
stretched by a given factor. This is formally defined below.

\begin{definition}
Let $d$ be a positive integer, $p$ be a positive rational, and $\omega \in \R^d$
be a vector such that $\omega(i) \ge 1$, for $1 \le i \le d$. The \emph{weighted
$L_{p,\omega}$- and $L_{\infty,\omega}$-norms} are defined respectively as
\[
||x||_{p,\omega} = \left(\textstyle\sum_{i=1}^d \omega(i) |x(i)|^p\right)^{1/p},
\quad\mbox{and}\quad
||x||_{\infty,\omega} = \textstyle\max_{1 \leq i \leq d} \omega(i)|x(i)|.
\]

Also, let $r \in \R_+$, and $c \in \R^d$. The \emph{$d$-dimensional $L_{p,
\omega}$-norm sphere} of radius $r$ and centered at $c$ is the set of points $x
\in \R^d$ such that $||x-c||_{p,\omega} < r$. The \emph{$d$-dimensional
$L_{\infty,\omega}$-norm sphere} is defined analogously.
\end{definition}

Figure~\ref{fig:spheres} contains examples of the considered spheres. Notice
that the assumption that $\omega(i)~\ge~1$, for $1 \le i \le d$, ensures that
each sphere of radius $r$ fits in a box with sides of length~$2r$. This is
without loss of generality, since we could normalize the radii of the
spheres otherwise. Also, we will consider only the case that $p \ge 1$. This
restriction implies that the spheres correspond to convex regions in the
$d$-dimensional space (see Figures~\ref{fig:spheres}(a) and \ref{fig:spheres}(f)
for examples when $p <1$).

% \MARGINFIXME{increase resolution of figure in the final version!}
\begin{figure}[htb!]
  \centering
  \def\samples{500}
  \def\w{.19\linewidth}
  \def\multiplier{2}
  \def\fillcolor{black!10}
  \def\lw{0.2pt}

  \newcommand{\subsphere}[1]{{\subfloat[]{\begin{minipage}[c]{\w}\centering{#1}\end{minipage}}}}

  \subsphere{\begin{tikzpicture}
    \def\f{(1 - sqrt(abs(\x)))^2}
    \draw[samples=\samples, smooth, line width=\lw, fill=\fillcolor] (-1,0) -- plot[domain=-1:1] (\x,{\f}) -- (1,0) -- plot[domain=1:-1] (\x,{-\f}) -- cycle;
  \end{tikzpicture}}
  \subsphere{\begin{tikzpicture}
    \draw[line width=\lw, fill=\fillcolor] (-1,0) -- (0,1) -- (1,0) -- (0,-1) -- cycle;
  \end{tikzpicture}}
  \subsphere{\begin{tikzpicture}
    \def\f{sqrt(1 - (\x)^2)}
    \draw[samples=\samples, smooth, line width=\lw, fill=\fillcolor] (-1,0) -- plot[domain=-1:1] (\x,{\f}) -- (1,0) -- plot[domain=1:-1] (\x,{-\f}) -- cycle;
  \end{tikzpicture}}
  \subsphere{\begin{tikzpicture}
    \def\f{sqrt(sqrt(1 - (\x)^4))}
    \draw[samples=\samples, smooth, line width=\lw, fill=\fillcolor] (-1,0) -- plot[domain=-1:1] (\x,{\f}) -- (1,0) -- plot[domain=1:-1] (\x,{-\f}) -- cycle;
  \end{tikzpicture}}
  \subsphere{\begin{tikzpicture}
    \def\f{1}
    \draw[samples=\samples, smooth, line width=\lw, fill=\fillcolor] (-1,0) -- plot[domain=-1:1] (\x,{\f}) -- (1,0) -- plot[domain=1:-1] (\x,{-\f}) -- cycle;
  \end{tikzpicture}}

  \subsphere{\begin{tikzpicture}
    \def\f{((1 - sqrt(abs(\x)))^2)/\multiplier}
    \draw[samples=\samples, smooth, line width=\lw, fill=\fillcolor] (-1,0) -- plot[domain=-1:1] (\x,{\f}) -- (1,0) -- plot[domain=1:-1] (\x,{-\f}) -- cycle;
  \end{tikzpicture}}
  \subsphere{\begin{tikzpicture}
    \draw[line width=\lw, fill=\fillcolor] (-1,0) -- (0,1.0/\multiplier) -- (1,0) -- (0,-1.0/\multiplier) -- cycle;
  \end{tikzpicture}}
  \subsphere{\begin{tikzpicture}
    \def\f{(sqrt(1 - (\x)^2))/\multiplier}
    \draw[samples=\samples, smooth, line width=\lw, fill=\fillcolor] (-1,0) -- plot[domain=-1:1] (\x,{\f}) -- (1,0) -- plot[domain=1:-1] (\x,{-\f}) -- cycle;
  \end{tikzpicture}}
  \subsphere{\begin{tikzpicture}
    \def\f{(sqrt(sqrt(1 - (\x)^4)))/\multiplier}
    \draw[samples=\samples, smooth, line width=\lw, fill=\fillcolor] (-1,0) -- plot[domain=-1:1] (\x,{\f}) -- (1,0) -- plot[domain=1:-1] (\x,{-\f}) -- cycle;
  \end{tikzpicture}}
  \subsphere{\begin{tikzpicture}
    \def\f{1.0/\multiplier}
    \draw[samples=\samples, smooth, line width=\lw, fill=\fillcolor] (-1,0) -- plot[domain=-1:1] (\x,{\f}) -- (1,0) -- plot[domain=1:-1] (\x,{-\f}) -- cycle;
  \end{tikzpicture}}

  \bigskip

  \begin{minipage}{0.8\linewidth}
      In order, the columns correspond to $p = 0.5, 1, 2, 4, \infty$,
      and the rows correspond to $\omega = (1,1), (1,2)$.
      All spheres have radius $1$.
  \end{minipage}
  \caption{Examples of $2$-dimensional $L_{p,\omega}$-spheres.}\label{fig:spheres}
\end{figure}

Now we show how to pack $d$-dimensional $L_{p,\omega} $-norm spheres, for
given $d$, $p$ and $\omega$. For that, we need three main steps:

\begin{enumerate}
  \item obtain an algorithm that decides whether a set of spheres can be
  packed in a given box, and provides a packing within an arbitrarily small
  error precision;

  \item show how to transform this approximate packing into a non-intersecting
  packing in an augmented bin;

  \item and show that the wasted volume after discarding bins that partially
  intersect spheres is a small factor of the spheres' volume.
\end{enumerate}

First, consider the case that $p$ is rational, and let $a$ and $b$ be positive
integers such that $p = a/b$. Also, let $v \in \R^d_+$ be the box size, and
$\omega\in \R^d_+$ be the norm weight. We want to obtain an algorithm similar to
that of Corollary~\ref{theorem:polysystem}. The convexity of the spheres implies
that deciding if a set $\mathcal{C} = \{1, \ldots, n\}$ of
$L_{\omega,p}$-spheres of radii $r_1, r_2, \ldots, r_n$ can be packed in a box
of size~$v$ can be encoded by the following system of inequalities, where
variables $x_1, x_2, \ldots, x_n \in \R^d$ correspond to centers of the spheres.

\begin{align}
  \textstyle\sum_{k = 1}^d (\omega(k)|x_i(k) - x_j(k)|)^p  \geq (r_i + r_j)^p &&& \mbox{for } 1 \leq i < j \leq n,\quad\mbox{and}\label{eq:sys1}\\
  r_i \leq x_i(k) \leq v(k) - r_i &&& \mbox{for } 1 \leq i \leq n,\; 1\leq k \leq d.
  \label{eq:sys2}
\end{align}

In order to transform~\eqref{eq:sys1}-\eqref{eq:sys2} into a polynomial system,
we have to deal with two problems: avoiding the modulus operator, and dealing
with non-integer exponents. Consider an expression~$e$. We can replace $|e|$ by
a new variable $z$, if we add the constraints $z^2 = e^2$, and $z \geq 0$. Also,
if $e \ge 0$, we can replace $e^{a/b}$ by a new variable $y$, by adding the
constraints $e^a = y^b$, and $y \ge 0$. Therefore, we can
transform~\eqref{eq:sys1}-\eqref{eq:sys2} into a system of polynomial
inequalities and equalities by adding new variables $y_{ij}$ and $z_{ij}$ in
$\R^d$ for every $1 \leq i < j \leq n$, and considering the following system.

\begin{align*}
  \textstyle\sum_{k = 1}^d y_{ij}(k)  \geq (r_i + r_j)^p &&& \mbox{for } 1 \leq i < j \leq n,\\
  r_i \leq x_i(k) \leq v(k) - r_i &&& \mbox{for } 1 \leq i \leq n, \; 1 \leq k \leq d,\\
  y_{ij}(k)^b = (\omega(k)z_{ik}(k))^a &&& \mbox{for } 1 \leq i < j \leq n, \; 1 \leq k \leq d,\\
  z_{ij}(k)^2 = (x_i(k) - x_j(k))^2 &&& \mbox{for } 1 \leq i < j \leq n, \; 1 \leq k \leq d,\quad\mbox{and}\\
  y_{ij}(k), z_{ij}(k) \geq 0 &&& \mbox{for } 1 \leq i < j \leq n, \; 1 \leq k \leq d.
\end{align*}

We can solve this system in a way similar to Corollary~\ref{theorem:polysystem},
and obtain an arrangement of the spheres that intersect by at most a small
value. Now, we modify this arrangement to obtain a feasible packing in an
augmented box, like in Lemma~\ref{lemma:rising}. The $d$-dimensional box will be
augmented only in its first dimension. Therefore, to avoid intersection of two
spheres, one of these will have the first coordinate shifted. The following
lemma gives a bound on the shift length that is necessary.

\begin{lemma}\label{lemma:rising2}
Let $r_1, r_2, h, \varepsilon$ be positive numbers such that $\varepsilon h \leq
r_1 + r_2 \leq h$, and $\varepsilon < 1$,  and  let $x_1$,~$x_2$ be points in
the $d$-dimensional space. If ${x_1(1) \geq x_2(1)}$, ${||x_1 - x_2||_{p,\omega}
\geq r_1 + r_2 - \varepsilon^2 h}$, and ${x_1' = (x_1(1) + t, x_1(2), \ldots,
x_1(d))}$, where $t = (2^a \varepsilon)^{1/p}h/\omega(1)$, then ${||x_1' -
x_2||_{p,\omega} \geq r_1 + r_2}$.
\end{lemma}

\begin{proof}
By direct calculation,
\begin{align*}
\textstyle
||x_1' - x_2||_{p,\omega}
    &\textstyle = \left(\sum_{k = 1}^d (\omega(k)|x_1'(k) - x_2(k)|)^p\right)^{1/p}\\
    &\textstyle = \left((\omega(1)|x_1'(1) - x_2(1)|)^p + \sum_{k = 2}^d (\omega(k)|x_1'(k) - x_2(k)|)^p\right)^{1/p}\\
    &\textstyle = \left((\omega(1)|x_1(1) + t - x_2(1)|)^p + \sum_{k = 2}^d (\omega(k)|x_1(k) - x_2(k)|)^p\right)^{1/p}\\
    &\textstyle = \left((\omega(1)|x_1(1) - x_2(1)| + \omega(1)t)^p + \sum_{k = 2}^d (\omega(k)|x_1(k) - x_2(k)|)^p\right)^{1/p}\\
    &\textstyle \geq \left((\omega(1)|x_1(1) - x_2(1)|)^p + t^p\omega(1)^p + \sum_{k = 2}^d (\omega(k)|x_1(k) - x_2(k)|)^p\right)^{1/p}\\
    &\textstyle = \left((||x_1 - x_2||_{p,\omega})^p + t^p\omega(1)^p\right)^{1/p}\\
    &\textstyle \geq \left((r_1 + r_2 - \varepsilon h)^p + t^p\omega(1)^p\right)^{1/p}\\
    &\textstyle \geq \left((r_1 + r_2 - \varepsilon(r_1 + r_2))^p + t^p\omega(1)^p\right)^{1/p}\\
    &\textstyle = \left((1-\varepsilon)^p(r_1 + r_2)^p + t^p\omega(1)^p\right)^{1/p}
\end{align*}

We now bound $t$ to obtain the desired result. First, notice that
because $1 - \varepsilon < 1$, we have that $1 - (1 - \varepsilon)^p \leq 1 - (1
- \varepsilon)^a$, and so
\begin{align*}
 1 - (1 - \varepsilon)^a
    &= 1 - \left(1 + \sum_{i = 1}^a \binom{a}{i}(-\varepsilon)^i\right)
    = - \sum_{i = 1}^a \binom{a}{i}(-\varepsilon)^i \\
    &\leq \sum_{i = 1}^a \binom{a}{i}\varepsilon^i
    \leq \varepsilon\sum_{i = 1}^a \binom{a}{i}
    < \varepsilon 2^a,
\end{align*}
from where we conclude that $1 - (1 - \varepsilon)^p < 2^a\varepsilon$. Now,
using the definition of $t$, we have that
\[
t = \frac{(2^a \varepsilon)^{1/p}h}{\omega(1)} \geq \frac{(1 - (1 -
\varepsilon)^p)^{1/p}h}{\omega(1)} = \left(\frac{(1 - (1 -
\varepsilon)^p)h^p}{\omega(1)^p}\right)^{1/p}.
\]
That is, $t^p\omega(1)^p \geq (1 - (1 - \varepsilon)^p)h^p \geq (1 - (1 -
\varepsilon)^p) (r_1 + r_2)^p$. Thus, we conclude that
\begin{align*}
||x_1' - x_2||_{p,\omega}
&\geq \left((1-\varepsilon)^p(r_1 + r_2)^p + t^p\omega(1)^p\right)^{1/p} \\
&\geq \left((1-\varepsilon)^p (r_1 + r_2)^p + (1 - (1 - \varepsilon)^p)(r_1 +
r_2)^p\right)^{1/p} = r_1 + r_2,
\end{align*}
and the results follows.
\myqed\end{proof}

For the case $p = \infty$, we cannot repeat the same lifting strategy to obtain
a packing in an augmented bin. As an example, even if two spheres intersect in a
small portion, it could be necessary to raise one of them by twice the radius of
the other sphere (see Figure~\ref{fig:square_intersection}). Therefore, instead
of using algebraic quantifier elimination algorithm to pack a set of spheres in
a given box, we can use an algorithm based on discretization. This has been
already done, for example, by Bansal~\emph{et~al.}~\cite{BansalCKS06} for
hypercubes, and the generalization is straightforward.

\begin{figure}[htb!]\centering
  \def\fillcolor{black!10}
  \def\lw{0.1pt}
  \begin{tikzpicture}[scale=1.5]
    \draw[line width=\lw, fill=\fillcolor] (0,0) -- (0,1) -- (1,1) -- (1,0) -- cycle;
    \draw[line width=\lw, fill=\fillcolor] (0.9,1) -- (0.9,2.5) -- (2.4,2.5) -- (2.4,1) -- cycle;
    \draw[dashed, line width=0.3pt] (0.9,0) -- (0.9,1.5) -- (2.4,1.5) -- (2.4,0) -- cycle;
    \node[circle, fill=black, inner sep=0pt, minimum size=3pt, outer sep=1pt, label={below: \tiny $(x,y)$}] (o) at (1.65,0.75) {};
    \node[circle, fill=black, inner sep=0pt, minimum size=3pt, outer sep=1pt, label={above: \tiny $(x,y+r_1)$}] (d) at (1.65, 1.75) {};
    \draw[->, >=stealth] (o) -- (d);
  \end{tikzpicture}

  \bigskip

  \begin{minipage}{0.8\linewidth}
  Fixing intersection of squares by lifting fails even if the intersection is
  small. In this example, one needs to raise one square by the length of its
  radius to avoid intersection. The square on the left has radius $0.5$, and the
  square on the right has radius $0.75$. To avoid intersection, we have to raise
  the square on the right by $1$.
  \end{minipage}

  \caption{An example where lifting fails.}\label{fig:square_intersection}
\end{figure}

Now, what is missing to extend Algorithm~1 is bounding the wasted volume of grid
elements that partially intersect spheres. The following lemma bounds the
distance between any two points in a grid element of side length $\ell$. Then,
we observe that the grid elements that are wasted must intersect the border of a
packed sphere.

\begin{lemma}\label{lemma:diaglp}
If $x$ and $y$ are two points in a hyperrectangle of side length $\ell$, then
\[
||x - y||_{p,\omega} \leq \ell (\textstyle\sum_{k = 1}^d \omega(k)^p )^{1/p},
\quad\mbox{and}\quad
||x - y||_{\infty,\omega} \leq \ell (\textstyle\max_{1 \leq k \leq d}\omega(k)).
\]
\end{lemma}
\begin{proof}
Note that, for every $1 \leq k \leq d$, $|x(i) - y(i)| \leq \ell$.
If $p \in \Q$, then we get
\begin{align*}
\textstyle
||x - y||_{p,\omega}
  &=    (\sum_{k = 1}^d (\omega(k) |x(i) - y(i)|)^p )^{1/p}\\
  &\leq (\sum_{k = 1}^d (\omega(k)\ell\         )^p )^{1/p}
   =    \ell (\sum_{k = 1}^d \omega(k)^p )^{1/p}.
\end{align*}
Now, if $p = \infty$, we get
\begin{align*}
\textstyle
||x - y||_{p,\omega}
  &=    \max_{1 \leq k \leq d}(\omega(k) |x(i) - y(i)|)\\
  &\leq \max_{1 \leq k \leq d}\omega(k) \ell
  =    \ell\max_{1 \leq k \leq d} \omega(k).\qedhere
\end{align*}
\end{proof}

\begin{lemma}\label{lemma:lpinterssphere}
Let $Q$ be a hyperrectangle of side length $\ell$, and define
\[
t = \begin{cases}
      \ell \, (\sum_{k = 1}^d \omega(k)^p )^{1/p}  & \mbox{if } p \in \Q,\\
      \ell \, (\max_{1 \leq k \leq d} \omega(k))   & \mbox{if } p = \infty.
    \end{cases}
\]
Also let $C$ be an $L_{p,\omega}$-norm sphere of radius $r$ centered at a point
$c$, such that $r \ge t$,  $C_+$ be the $L_{p,\omega}$-norm sphere of radius ${r
+ t}$ centered at $c$, and $C_-$ be the interior of the $L_{p,\omega}$-norm
sphere of radius ${r - t}$ centered at $c$. If $C$ intersects $Q$, but $Q$ is
not contained in $C$, then $Q \subseteq C_+$ and $Q \cap C_- = \emptyset$.
\end{lemma}

\begin{figure}[t]
  \centering
  \begin{tikzpicture}[scale = 1.5]
    \draw (0,0) circle (1);
    \node[circle, fill=black, inner sep=0pt, minimum size=3pt, label={225:\small $x$}] (x) at ({sqrt(2)/2}, {sqrt(2)/2}) {};
    \node[circle, fill=black, inner sep=0pt, minimum size=3pt, label={225:\small $y$}] (y) at ($(x)-(0.3,0.3)$) {};
    \node[circle, fill=black, inner sep=0pt, minimum size=3pt, label={225:\small $z$}] (z) at ($(x)+(0.3,0.3)$) {};
    \draw[dashed] ($(x)-(0.3,0.3)$) rectangle ($(x)+(0.3,0.3)$);
    \draw[dotted] (0,0) circle ({1 - 0.3*sqrt(2)});
    \draw[dotted] (0,0) circle ({1 + 0.3*sqrt(2)});
    \node at ({1 - 0.3*sqrt(2) + 0.15},0) {$C_-$};
    \node at ({1 + 0.09},0) {\small $C$};
    \node at ({1 + 0.3*sqrt(2) + 0.15},0) {\small $C_+$};
  \end{tikzpicture}

  \bigskip

  \begin{minipage}{0.8\linewidth}
Point $x$ is at the boundary of $C$. The dashed box has side $2\ell$, and is
centered at $x$. Point $y$ is at a distance $\ell\sqrt{2}$ of $x$, and is not
contained in $C_-$. Similarly, $z$ is not contained in $C_+$. For this
configuration, every box of side $\ell$ that intersects $x$ must be contained in
the region delimited by the dashed square.
  \end{minipage}
  \caption{Bounding the box intersection to a sphere boundary.}\label{fig:circle-square}
\end{figure}
\begin{proof}
Since $C \cap Q \neq \emptyset$, and $Q \not\subseteq  C$, there exists $x \in
Q$ such that $||c - x||_{p,\omega} = r$. Also, consider an arbitrary point $y
\in Q$ (see Figure~\ref{fig:circle-square}). Using the triangle inequality, and
Lemma~\ref{lemma:diaglp}, we have that
\[
||c - y||_{p,\omega} \leq ||c - x||_{p,\omega} + ||x - y||_{p,\omega} \leq r + t,
\]
from where we conclude that $y \in C_+$, and hence $Q \subseteq C_+$. Again from
the triangle inequality, and Lemma~\ref{lemma:diaglp}, we get
\[
r = ||c - x||_{p,\omega} \leq ||c - y||_{p,\omega} + ||x - y||_{p,\omega} \leq ||c - y||_{p,\omega} + t,\]
so $||c - y||_{p,\omega} \geq r - t$, and we conclude that $y \notin C_-$.
Thus,  we get $Q \cap C_- = \emptyset$.
\myqed\end{proof}

From Lemma~\ref{lemma:lpinterssphere}, one may obtain a statement similar to
Lemma~\ref{lemma:area}. Using arguments of Section~\ref{sec:aptas}, combined
with the discussion in this subsection, it is now possible to obtain an
asymptotic approximation scheme for the problem of packing
$L_{\omega,p}$-spheres in the minimum number of $d$-dimensional bins of a given
size. It is straightforward to extend these results for the strip packing
problem of $L_{\omega,p}$-spheres, when the recipient contains an unbounded
dimension, and the objective is to minimize span of packed spheres in this
dimension.

We finish this subsection by noticing that the bin packing of
$L_{\infty,\omega}$-norm spheres is, in fact, a particular case of the general
hyperrectangle bin packing problem when all hyperrectangles are congruent.

\subsection{Generalizing the bin}

In this subsection, we sketch how our results can be generalized to deal with
bins of different shapes, that are described by semi-algebraic sets. This
generalization includes, for example, polytopes and $L_{p,\omega}$-norm spheres.
The idea is solving the corresponding system of polynomial inequalities,
together with the packing constraints.
We will consider a bin $B$ that corresponds to the $d$-dimensional
semi-algebraic set described by polynomials $f_1, \ldots, f_s \in \Q[z_1,
\ldots, z_d]$. More precisely, a point $z \in \R^d$ is contained in $B$ if, and
only if, $f_i(z) \geq 0$ for all $1 \le i \le s$. In the following, we will
derive a quantified formula with polynomial inequalities and equalities to
decide if a set $\mathcal{C}$ of spheres can be packed in $B$.

We begin by defining variables $x_1, \ldots, x_n \in \R^d$ to represent the
centers of the spheres. Consider the inequalities
\begin{align}
  ||x_i - x_j||_{p,\omega} \geq r_i + r_j, &&& \mbox{ for } 1 \leq i < j \leq n, \label{ineq1} \\
  f_j(x_i) \geq 0, &&& \mbox{ for } 1 \leq i \leq n, \; 1 \leq j \leq s.\label{ineq2}
\end{align}
These constraints requires that there exists a feasible arrangement of the
spheres in the \mbox{$d$-dimensional} space, and that each center is contained
in the bin. Now, we will ensure that every point of the sphere is indeed in the
bin. To do this, for each sphere $i$, we will consider a mapping from the
$(d+1)$-dimensional space to this sphere, centered at $x_i$, that is, the
function $g_i : \R^d \setminus \{0\} \times \R \rightarrow \R^d$ defined as
\[
  g_i(a, \lambda) = r_i\frac{a}{||a||_{p,\omega}(1 + \lambda^2)} + x_i.
\]
This maps any pair $(a,\lambda)$ of the domain to a point $y_i = g_i(a,\lambda)$
in the $L_{p,\omega}$-sphere centered at~$x_i$ with radius $r_i$, or at its
boundary. Analogously, each point in the sphere, with exception of the center, is
mapped to a point of the domain. Therefore, it is enough to make sure that every
such mapped point $y_i$ is in the bin. This is done by adding the following
constraints.
\begin{align}
  ||a||_{p,\omega}(1 + \lambda^2)\left(y_i(k) - x_i(k)\right) = r_i a(k), &&& \mbox{ for } 1 \leq i \leq n, \; 1 \leq k \leq d, \label{ineq3}\\
  f_j(y_i) \geq 0, &&& \mbox{ for } 1 \leq i \leq n, \; 1 \leq j \leq s. \label{ineq4}
\end{align}
Notice that if $||a||_{p,\omega} = 0$, then inequality~\eqref{ineq3} is
trivially satisfied, since in this case we have~${a(k) = 0}$ for every $1 \leq k
\leq d$, so $y_i$ can be set as any point of the bin to satisfy
inequality~\eqref{ineq4}.
Now, conjoining all the constraints, we obtain a quantified first-order formula.
Therefore, deciding if there is a packing of $\mathcal{C}$ in $B$ is equivalent
to deciding the formula
\[
(\exists x_1 \in \R^d)  \dots (\exists x_n \in \R^d)
(\forall a \in \R^d) (\forall \lambda \in \R)
(\exists y_1 \in \R^d)  \dots (\exists y_n \in \R^d)
\;
\mbox{\eqref{ineq1}-\eqref{ineq4}}.
\]

Although constraints~\eqref{ineq1} and~\eqref{ineq3} contain non-polynomial
terms, it is possible to obtain an equivalent formula with only polynomial
inequalities and equalities using the process discussed before. Therefore, one
can use any quantifier elimination algorithm to solve it, provided that it
supplies a realization of the points at a given precision.

Also in this case, the obtained center coordinates can be irrational, so only
approximate values may be available. To use resource augmentation, here we will
scale both the recipient and the radii of the spheres by a factor $1 + \gamma$
before solving the algebraic system. This is done by replacing polynomial $f_i$
by polynomial $f_i' = f_i(x/(1+\gamma))$, for $1 \le i \le s$, and radius $r_i$
by $r_i' = (1+\gamma)r_i$, for $1 \le i \le n$. After obtaining an approximate
packing, we scale back the radii of the spheres to their original values. The
precision of the obtained packing is adjusted according to parameter~$\gamma$,
so that intersection is avoided.

Notice that after the first iteration of the packing algorithm, the remaining
space of the bin will be partitioned into a grid of boxes. Therefore, the volume
of the boxes that intersect the boundary of the bin will be lost. This is not a
problem because we use a fine-grained grid, for which the wasted volume will be
bounded by a small fraction of the bin volume.

\section{Final Remarks}\label{sec:conclusion}

We presented the first approximation algorithms for the circle bin packing
problem using augmented bins, and the circle strip packing problem. We obtained
asymptotic approximation schemes for circle packings exploring novel ideas, such
as iteratively distinguishing large and small items, and carefully using the
free space left after packing large items. We believe that our algorithm can
lead to further results for related problems, and we have already presented some
possible generalizations. Also, our use of algebraic quantifier elimination
algorithms exemplifies how results from algebra can be successfully used in the
context of optimization. Using these algorithms helped us to avoid
discretization algorithms, whose running time would depend exponentially on the
size of resource augmentation parameter,  $\log 1/\gamma$, and allowed the
packing of more general items, in a simple and concise manner.

We highlight that our algorithm is not restricted to circles, and indeed it can
be seen as a unified framework to different packing problems. These problems
need to satisfy certain conditions, such as requiring that items may be
partitioned into groups of large and small items (as in
Figure~\ref{fig:partition}), and the packing of large items in an augmented bin
can be approximately solved. One may consider, for example, the problem of
packing regular polygons, by using a discretization algorithm to deal with large
items. Notice that one could even consider instances with items of mixed shapes.
A very natural generalization of our algorithm is considering the packing of
$L_{p}$-norm spheres, as done in Section~\ref{sec:generalizations}. Moreover,
minor modifications of the algorithm allow using bins of different shapes. In
Section~\ref{sec:generalizations}, the illustrative example considers a whole
set of semi-algebraic sets as possible bins, what comprises, for instance, the
problem of packing spheres in spheres. For the particular case that a bin is
enlarged by a constant, we avoid the approximation, and obtain a resource
augmentation scheme.

Finally, we note that, although the quantifier elimination algorithms we used
give a precise representation of a packing in a non-augmented bin, the returned
solution may possibly contain irrational coordinates. To provide solutions with
rational numbers, we used approximate coordinates with arbitrary precision. This
is the \emph{only} reason why we used augmented bins in our APTAS, and thus
resource augmentation can be avoided in a more general computational model. We
left open the question to determine if it is always possible to obtain a
rational solution to the problem of packing a set of circles of rational radii
in a non-augmented bin of rational dimensions.

\paragraph{Acknowledgement}
We would like to thank Frank Vallentin for providing us with insights and
references on the cylindric algebraic decomposition and other algebraic notions.

\bibliographystyle{abbrv}
\bibliography{circle}

\begin{thebibliography}{10}

\bibitem{BakerCR80}
B.~Baker, E.~Coffman, Jr., and R.~Rivest.
\newblock Orthogonal packings in two dimensions.
\newblock {\em SIAM Journal on Computing}, 9(4):846--855, 1980.

\bibitem{BansalCS10}
N.~Bansal, A.~Caprara, and M.~Sviridenko.
\newblock {A New Approximation Method for Set Covering Problems, with
  Applications to Multidimensional Bin Packing}.
\newblock {\em SIAM Journal on Computing}, 39(4):1256--1278, 2010.

\bibitem{BansalCKS06}
N.~Bansal, J.~R. Correa, C.~Kenyon, and M.~Sviridenko.
\newblock {Bin Packing in Multiple Dimensions: Inapproximability Results and
  Approximation Schemes}.
\newblock {\em Mathematics of Operations Research}, 31(1):31--49, 2006.

\bibitem{BansalHISZ13}
N.~Bansal, X.~Han, K.~Iwama, M.~Sviridenko, and G.~Zhang.
\newblock A harmonic algorithm for the {3D} strip packing problem.
\newblock {\em SIAM Journal on Computing}, 42(2):579--592, 2013.

\bibitem{BansalK14}
N.~Bansal and A.~Khan.
\newblock {Improved Approximation Algorithm for Two-Dimensional Bin Packing}.
\newblock In {\em Proceedings of the Twenty-Fifth Annual ACM-SIAM Symposium on
  Discrete Algorithms}, pages 13--25, 2014.

\bibitem{BasuPR96}
S.~Basu, R.~Pollack, and M.-F. Roy.
\newblock {On the Combinatorial and Algebraic Complexity of Quantifier
  Elimination}.
\newblock {\em Journal of the ACM}, 43(6):1002--1045, Nov. 1996.

\bibitem{BasuPR06}
S.~Basu, R.~Pollack, and M.-F. Roy.
\newblock {\em Algorithms in Real Algebraic Geometry}.
\newblock Springer-Verlag Berlin Heidelberg, 2006.

\bibitem{BirginG10}
E.~G. Birgin and J.~M. Gentil.
\newblock {New and improved results for packing identical unitary radius
  circles within triangles, rectangles and strips}.
\newblock {\em Computers and Operations Research}, 37(7):1318--1327, 2010.

\bibitem{Caprara02}
A.~Caprara.
\newblock Packing 2-dimensional bins in harmony.
\newblock In {\em Proceedings of the 43rd Annual IEEE Symposium on Foundations
  of Computer Science}, pages 490--499, 2002.

\bibitem{ChungGJ82}
F.~Chung, M.~Garey, and D.~Johnson.
\newblock {On Packing Two-Dimensional Bins}.
\newblock {\em SIAM Journal on Algebraic Discrete Methods}, 3(1):66--76, 1982.

\bibitem{CoffmanGJT80}
E.~Coffman, Jr., M.~Garey, D.~Johnson, and R.~Tarjan.
\newblock Performance bounds for level-oriented two-dimensional packing
  algorithms.
\newblock {\em SIAM Journal on Computing}, 9(4):808--826, 1980.

\bibitem{CoffmanCGMV13}
J.~E.~G. Coffman, J.~Csirik, G.~Galambos, S.~Martello, and D.~Vigo.
\newblock {Bin Packing Approximation Algorithms: Survey and Classification}.
\newblock In P.~M. Pardalos, D.-Z. Du, and R.~L. Graham, editors, {\em
  {Handbook of Combinatorial Optimization}}, pages 455--531. Springer New York,
  2013.

\bibitem{Collins75}
G.~E. Collins.
\newblock {Quantifier elimination for real closed fields by cylindrical
  algebraic decompostion}.
\newblock In {\em {Proceedings of the 2nd GI Conference on Automata Theory and
  Formal Languages}}, pages 134--183, 1975.

\bibitem{DemaineFL10}
E.~D. Demaine, S.~P. Fekete, and R.~J. Lang.
\newblock {Circle Packing for Origami Design Is Hard}.
\newblock In {\em Proceedings of the 5th International Conference on Origami in
  Science}, pages 609--626, 2010.

\bibitem{Eisenbrand03}
F.~Eisenbrand.
\newblock Fast integer programming in fixed dimension.
\newblock In G.~Di~Battista and U.~Zwick, editors, {\em Algorithms - ESA 2003},
  volume 2832 of {\em Lecture Notes in Computer Science}, pages 196--207.
  Springer Berlin Heidelberg, 2003.

\bibitem{FernandezdelaVegaL81}
W.~{Fernandez de la Vega} and G.~S. Lueker.
\newblock {Bin packing can be solved within \hbox{1 + $\varepsilon$} in linear
  time}.
\newblock {\em Combinatorica}, 1(4):349--355, 1981.

\bibitem{GeorgeGL95}
J.~A. George, J.~M. George, and B.~W. Lamar.
\newblock {Packing different-sized circles into a rectangular container}.
\newblock {\em European Journal of Operational Research}, 84(3):693--712, 1995.

\bibitem{GrigorevV88}
D.~Y. Grigor'ev and N.~N. {Vorobjov Jr}.
\newblock {Solving systems of polynomial inequalities in subexponential time}.
\newblock {\em Journal of Symbolic Computation}, 5(1--2):37--64, 1988.

\bibitem{HarrenJPS11}
R.~Harren, K.~Jansen, L.~Prädel, and R.~van Stee.
\newblock A $(5/3+\varepsilon)$-approximation for strip packing.
\newblock In {\em Proceedings of the 12th Algorithms and Data Structures
  Symposium}, pages 475--487, 2011.

\bibitem{HarrenS09}
R.~Harren and R.~van Stee.
\newblock Improved absolute approximation ratios for two-dimensional packing
  problems.
\newblock In {\em Proceedings of the 12th International Workshop, APPROX 2009,
  and 13th International Workshop, RANDOM 2009}, pages 177--189, 2009.

\bibitem{HarrenS12}
R.~Harren and R.~van Stee.
\newblock Absolute approximation ratios for packing rectangles into bins.
\newblock {\em Journal of Scheduling}, 15(1):63--75, 2012.

\bibitem{HifiM09}
M.~Hifi and R.~M'Hallah.
\newblock A literature review on circle and sphere packing problems: Models and
  methodologies.
\newblock {\em Advances in Operations Research}, 2009:1--22, 2009.

\bibitem{JansenS06}
K.~Jansen and R.~Solis-Oba.
\newblock An asymptotic approximation algorithm for {3D}-strip packing.
\newblock In {\em Proceedings of the Seventeenth Annual ACM-SIAM Symposium on
  Discrete Algorithm}, pages 143--152, 2006.

\bibitem{KenyonR00}
C.~Kenyon and E.~R{\'e}mila.
\newblock {A Near-Optimal Solution to a Two-Dimensional Cutting Stock Problem}.
\newblock {\em Mathematics of Operations Research}, 25(4):645--656, 2000.

\bibitem{KohayakawaMRW04}
Y.~Kohayakawa, F.~Miyazawa, P.~Raghavan, and Y.~Wakabayashi.
\newblock {Multidimensional Cube Packing}.
\newblock {\em Algorithmica}, 40(3):173--187, 2004.

\bibitem{Lenstra83}
J.~Lenstra, H.~W.
\newblock Integer programming with a fixed number of variables.
\newblock {\em Mathematics of Operations Research}, 8(4):pp.538--548, 1983.

\bibitem{LiC90}
K.~Li and K.~Cheng.
\newblock On three-dimensional packing.
\newblock {\em SIAM Journal on Computing}, 19(5):847--867, 1990.

\bibitem{MeirM68}
A.~Meir and L.~Moser.
\newblock {On packing of squares and cubes}.
\newblock {\em Journal of Combinatorial Theory}, 5(2):126--134, 1968.

\bibitem{MiyazawaW97}
F.~Miyazawa and Y.~Wakabayashi.
\newblock An algorithm for the three-dimensional packing problem with
  asymptotic performance analysis.
\newblock {\em Algorithmica}, 18(1):122--144, 1997.

\bibitem{Schiermeyer94}
I.~Schiermeyer.
\newblock Reverse-fit: A 2-optimal algorithm for packing rectangles.
\newblock In {\em Proceedings of the Second Annual European Symposium on
  Algorithms}, pages 290--299, 1994.

\bibitem{Sleator80}
D.~D. Sleator.
\newblock A 2.5 times optimal algorithm for packing in two dimensions.
\newblock {\em Information Processing Letters}, 10(1):37--40, 1980.

\bibitem{Steinberg97}
A.~Steinberg.
\newblock A strip-packing algorithm with absolute performance bound 2.
\newblock {\em SIAM Journal on Computing}, 26(2):401--409, 1997.

\bibitem{SzaboMCSCG07}
P.~G. Szab{\'o}, M.~C. Mark{\'o}t, T.~Csendes, E.~Specht, L.~Casado, and
  I.~Garc{\'i}a.
\newblock {\em {New Approaches to Circle Packing in a Square}}.
\newblock Springer, 2007.

\bibitem{Tarski51}
A.~Tarski.
\newblock {\em {A decision method for elementary algebra and geometry}}.
\newblock University of California Press, 1951.

\bibitem{WaescherHS2007}
G.~Wäscher, H.~Haußner, and H.~Schumann.
\newblock An improved typology of cutting and packing problems.
\newblock {\em European Journal of Operational Research}, 183(3):1109--1130,
  2007.

\end{thebibliography}

\end{document}